\documentclass[aps, pra, english,superscriptaddress,floatfix,twocolumn]{revtex4-2}

\usepackage{graphicx}% Include figure files
\usepackage{dcolumn}% Align table columns on decimal point
\usepackage{bm}% bold math
\usepackage{braket}
\usepackage{amsmath, amsthm, amssymb}
\usepackage{color}
\usepackage[english]{babel}
\usepackage{caption, subcaption}	
\usepackage{qcircuit}
\usepackage{array}
\newcommand{\mytilde}{\raise.17ex\hbox{$\scriptstyle\mathtt{\sim}$}}
\newtheorem{theorem}{Theorem}
\newtheorem{lemma}[theorem]{Lemma}

\theoremstyle{definition}
\newtheorem{definition}{Definition}

\theoremstyle{definition}

\makeatletter
\DeclareRobustCommand\ket[1]{%
  \@ifnextchar\bra{\k@t{#1}\!}{\k@t{#1}}%
}
\newcommand\k@t[1]{{|{#1}\rangle}}
\makeatother

\begin{document}

\preprint{APS/123-QED}

\title{Distributed Quantum Error Correction with \\Permutation-Invariant Approximate Codes}

\author{Connor Clayton}
    \affiliation{Leidos, Inc.,
        4001 Fairfax Dr, Arlington, VA 22203}
    \affiliation{Joint Center for Quantum Information and Computer Science, 
        University of Maryland, College Park, MD 20742, USA}

\author{Bruno Avritzer}
    \affiliation{Leidos, Inc.,
        4001 Fairfax Dr, Arlington, VA 22203}
    \email{Bruno.Avritzer@leidos.com}

\date{\today}


\begin{abstract}

Modular quantum computing architectures require error correction schemes that remain effective in the presence of noisy inter-processor operations.
We introduce a distributed quantum error correction framework based on approximate codes to address this challenge.
Our approach enables concatenation of distinct local codes across modules while allowing logical operations composed primarily of processor-local gates.
We derive a lower bound on processor-nonlocal gate reduction in a distributed context and present corresponding simulations which indicate that this nontraditional approach can provide marked advantage over existing approaches in the highly non-uniform error landscape of a distributed quantum computer.
As a concrete realization, we present encoding and decoding circuits for the permutation-invariant W-state code and propose efficient methods for its preparation.
These results highlight the potential of approximate distributed error correction strategies for scalable, modular, fault-tolerant quantum computation.

\end{abstract}

\maketitle

\section{Introduction}

Modular quantum computing architectures, composed of multiple chips or processors, have emerged as a practical route to scale up qubit counts beyond the engineering limitations of a single device~\cite{distributed-qc-survey}. In such architectures, quantum processing units (QPUs) are networked via interconnects (typically optical fiber~\cite{distributed-qc-optical-link} or microwave~\cite{microwave-interconnect} links) to act as one larger computer. This approach is being pursued across platforms including superconducting circuits~\cite{ibm-combining-processors-classical, tour-de-gross}, trapped ion arrays~\cite{distributed-qc-optical-link, modular-ionq, modular-trapped-ion}, and photonics~\cite{scaling-modular-photonic-qc}.

However, modular quantum computing faces its own challenge: inter-module operations are noisy and slow, exacerbating the already non-trivial task of practical quantum error correction (QEC). Traditional QEC codes like the surface code assume a monolithic array of qubits with uniform error rates and fast local gates, conditions that modular systems violate. Distributed gates suffer from photon loss, interconnect latency, and other sources of error during inter-module communication.

To date, two broad approaches to quantum error correction in modular architectures have been studied: \textbf{local QEC (LQEC)}, where each module independently runs an error-correcting code and only high-level quantum operations (teleportations, remote gates) connect the modules, and \textbf{distributed QEC (DQEC)}, where each individual QEC code block is spread across multiple modules. While LQEC is the most commonly adopted approach, DQEC has recently emerged as a promising candidate for better performing code layouts in distributed systems~\cite{buffalo-distributed, cisco-distributed-color-code, elkouss-dqec, nu-quantum-floquet, jiang-cosmic-ray-errors}.
The most appealing quality of DQEC is its capacity to reduce the number of non-local gates in an error-corrected circuit by housing physical qubits from more code blocks on each processor, allowing transversal gates to be executed locally -- however, this must be balanced against a decrease in locality for error correction decoding.
In addition, it has previously been found that DQEC enables greater resistance to local errors due to catastrophic events~\cite{jiang-cosmic-ray-errors} and we extend this characteristic to more general spatially-correlated errors.
%We find that DQEC has two major advantages over LQEC: a better error correction performance and utility for evading losses due to distribution (as opposed to a single processor that contained the same number of qubits locally). 

Existing formulations of DQEC, however, face a fundamental roadblock to the practical realization of this opportunity: non-local gate reduction via DQEC is most effectively achieved by transversal gates, such that all physical gates may be executed locally on each processor, %(despite the code block being distributed)
but the Eastin-Knill theorem~\cite{eastin-knill} precludes any code from exactly correcting errors with a universal set of transversal gates.
Existing DQEC schemes therefore either forfeit gate locality~\cite{elkouss-dqec} or execute expensive alternatives such as magic state injection to implement universal quantum computation~\cite{buffalo-distributed, cisco-distributed-color-code, nu-quantum-floquet}.

%The basic idea of DQEC is that there is a trade-off: processor-nonlocal gates are required to encode and decode into a distributed code, but are not required for processing on the encoded qubits (assuming the gates in question can be performed transversally within that code). If the lifetime of the logical qubit is not a limitation, this may reduce the total number of expensive processor-nonlocal gates required to execute a logical quantum circuit. However, due to the Eastin-Knill theorem~\cite{eastin-knill}, it is impossible to perform a set of universal operations on the encoded state transversally. There are many ways around this limitation, the most common of which is magic state distillation~\cite{bravyi-kitaev-distillation}. This introduces a significant resource overhead, which is exacerbated by poor nonlocal gate fidelity in the distributed case. 

In this paper, we argue for a new paradigm of \emph{Distributed Approximate Quantum Error Correction (DAQEC)}, which integrates approximate QEC codes into distributed, modular architectures.
Approximate codes sacrifice exact error correction but are not subject to the Eastin-Knill theorem, and therefore may admit a universal transversal gate set~\cite{faist}. %Although there are other approaches for 
We find that DAQEC can outperform both exact distributed QEC and conventional local QEC in a modular quantum computing context by leveraging inhomogeneous error conditions and enabling universal transversal quantum computation. %; a summary of these results can be found in Table~\ref{overview-table}.
We also find that such schemes admit  additional desirable properties which facilitate the composition and decoding of codes across processors.

This paper makes the following contributions:
\begin{itemize}
\item \textbf{Distributed Approximate Quantum Error Correction (DAQEC).} We introduce a novel DAQEC framework and derive a lower bound on its performance advantage in specific regimes.
\item \textbf{Distributed QEC Performance Advantage via Non-Local Gate Reduction.} We quantify the reduction in processor-non-local gates through DQEC and demonstrate its potential performance advantage through detailed simulation. We further provide an interconnect-fidelity-independent bound on DQEC advantage.
\item \textbf{Improved W-state Preparation Circuits.} We present a more resource-efficient construction for creating $2^k$-qudit W states, improving on~\cite{yeh-qudit-w-state} for these instances.
\item \textbf{Explicit W-state Code Circuits.} We present the first explicit constructions of the encoding and decoding circuits of the W-state code, which may be of independent interest beyond our results on distributed error correction.
\item \textbf{Collective Fault-Tolerant Paradigm.} We introduce the concept of collective fault-tolerance, in which a network of fault-tolerant processors forms a collectively fault-tolerant system via distributed code concatenation, and show how this can be achieved using DAQEC.

\end{itemize}

The remainder of this paper is structured as follows. In Section~\ref{sec:dqec}, we examine the tradeoffs that come with distributing code blocks and compare our approach to related work.
In Section~\ref{sec:w-state-code}, we discuss the W-state code and give its explicit encoding and decoding circuits.
Then in Section~\ref{adqec}, we detail the advantages of DAQEC in modular systems.
Finally, we discuss opportunities for further research in Section~\ref{discussion}.

\section{Quantum Error Correction in the Multi-QPU Setting}
\label{sec:dqec}
Elevated noise levels and limited connectivity across module interconnects in the distributed setting introduce novel challenges in quantum error correction beyond those present in a monolithic device.
A fundamental question in this context is how to allocate code blocks to physical systems, which may each be local to a processor, or may be partially or entirely distributed.
As we will see, the arrangement of blocks within or between processors brings about a tradeoff between the complexities of two-qubit logical gates and QECC decoding, carrying implications for the system's capacity to correct errors.

\subsection{Local QEC}

Local QEC deploys code blocks module-wise (Figure~\ref{overview}, left). In this straightforward approach, each QPU module protects its qubits with an internal QEC code (for example, a surface code on a 2D grid of superconducting qubits, or a Bacon-Shor code on an ion trap device). Quantum communication between modules is then required only for logical gates between logical qubits on different processors. For example, to perform a CNOT between logical qubits in the green and purple blocks of Figure~\ref{overview} (left), one must perform remote gates or teleportation via shared entanglement.

The LQEC approach has the benefit that well-developed QEC codes and decoders for single-processor systems can be applied with minimal modification on each module. Indeed, early demonstrations of logical qubits have been achieved on single devices~\cite{bluvstein2024logical, da2404demonstration, sivak2023real} and recently have been shown to correct errors below threshold~\cite{google-qec}. By networking such devices, one can in principle build a larger, scalable quantum computer, similar in concept to the way classical supercomputers scale by connecting many smaller processors. Recently, a proof of concept modular architecture for photonic platforms has been demonstrated~\cite{scaling-modular-photonic-qc}.

However, the limitations of local QEC become apparent as the number of modules grows. Logical gates must be executed between processors, severely hampering gate fidelity under noisy interconnect conditions and putting strain on the already resource-intensive inter-module entanglement generation system. 
Intuitively, LQEC is a natural approach which induces modules which each fully contain a set of logical qubits but which must rely on interconnects to execute gates between remote logical qubits.
This increases sensitivity to link errors, which in turn can drastically lower the effective fault-tolerant threshold of the whole machine.
% DQEC facilitates the fault-tolerant implementation of logical gates via processor-local physical gates.

\begin{figure}
\includegraphics[width=\columnwidth]{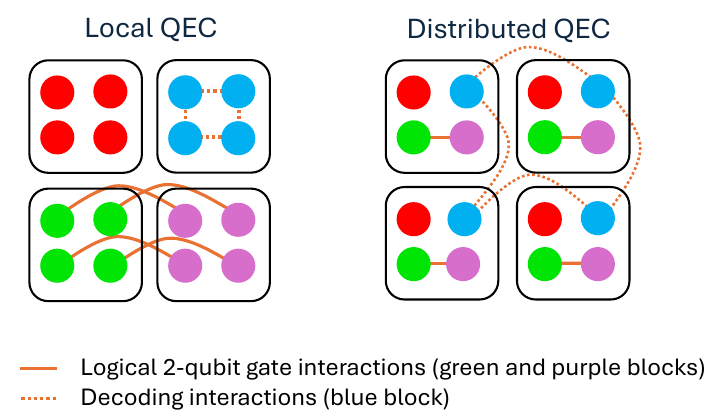}
\centering
\caption{Comparison of interactions required for local QEC and distributed QEC in a 4-QPU system. Each QPU holds four qubits and qubits of the same color belong to a single code block. Distributed QEC enables logical 2-qubit gates comprised entirely of local physical gates (assuming transversality), at the expense of a more complex distributed decoding procedure.}
\label{overview}
\end{figure}

\subsection{Distributed QEC}

Limitations of local QEC motivate \emph{distributed} QEC, an underexplored approach which encodes logical qubits into physical systems distributed across multiple modules. Below, we discuss how employing distributed code blocks has the potential to reduce processor interconnect utilization and improve resilience to correlated errors. We acknowledge that with a DQEC approach, one ostensibly sacrifices the efficiency of standard decoding strategies which are readily applicable in LQEC (see Figure~\ref{overview}, right), motivating a new area of research into decoding algorithms under processor-biased noise and with heterogeneous interconnect efficiencies. In Section~\ref{adqec}, we will formulate an approach using approximate codes to mitigate the difficulties of DQEC.

\subsubsection{Fewer Processor-Non-Local Gates}

\begin{figure} 
\centering
\begin{subfigure}{\columnwidth}
        \includegraphics[width=\columnwidth]{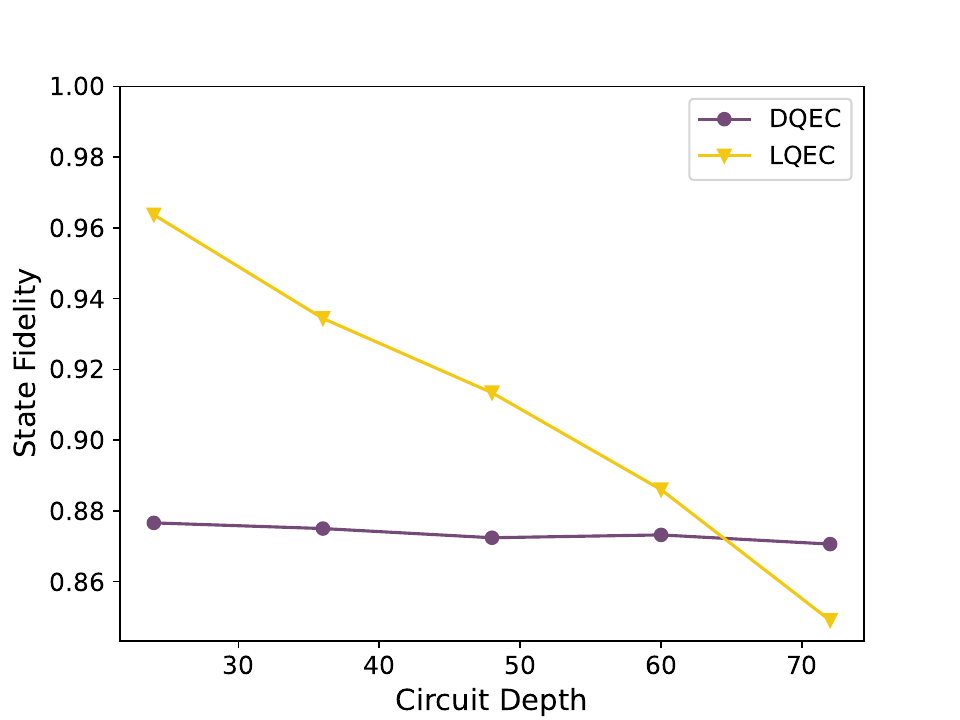}
        \caption{LQEC requires noisy inter-processor gates for logical operations and is outperformed by DQEC at sufficient depth between decoding operations.}
    \end{subfigure}    
    \begin{subfigure}{\columnwidth}
        \includegraphics[width=\columnwidth]{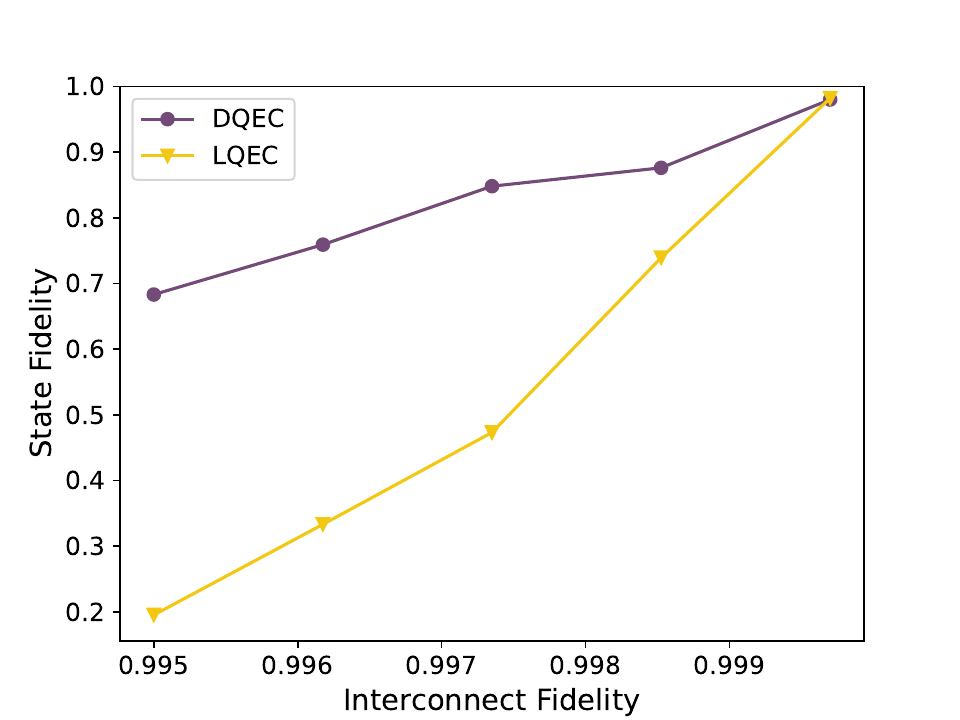}
        \caption{For sufficient depth between decoding stages, DQEC leads to higher final state fidelity for sub-optimal interconnect fidelities.}
    \end{subfigure}
\caption{
Relative performance of distributed and local code block allocation schemes as a function of circuit depth between decoding stages.
Note that these results apply to the 7-qubit Steane code only for circuits composed of transversal gates; a method applicable to general circuits is presented in Section~\ref{adqec}.
}
\label{fig:pnl-reduction}
\end{figure}

A primary advantage of distributing code blocks across processors is that it provides the opportunity to execute transversal logical gates locally (within a single processor), since this arrangement allows physical qubits from a larger number of code blocks to be placed onto each QPU. 
By optimizing the allocation of physical qubits to code blocks, one may minimize or eliminate the need for remote physical gates during logical gate execution, thereby reducing the infidelity associated with these inherently noisy operations.% It is also possible to allocate logical blocks within a processor in such a way as to minimize the utilization of noisier local connections during transversal gate execution.

This distributed layout introduces a tradeoff, as distributed decoding must be performed instead. However, for sufficient logical circuit depth between each decoding gadget and for sufficiently shallow decoding circuits, the total number of non-local gates is lower under distributed error correction.
In this way, DQEC techniques provide a viable path toward resource reduction in modular quantum computing architectures, and will become more advantageous relative to LQEC as the error rates of individual components improve.

We validate this behavior through detailed simulations of the relative performance of LQEC and DQEC. The results are shown in Figure~\ref{fig:pnl-reduction}.
We simulate seven $[[7,1,3]]$ Steane code blocks across seven processors, each hosting 13 qubits. In the LQEC configuration, all seven encoded qubits and the six required ancillas reside within a single processor. In contrast, in the DQEC configuration, both the encoded qubits and the ancillas are fully distributed across processors.
Each two-qubit gate is followed by a depolarizing channel with error parameter $p=2\times10^{-3}$ for remote gates and $p=2\times10^{-4}$ for processor-local gates.
While these error rates exceed current hardware capabilities, they represent the regime of practical interest for demonstrating distributed quantum computation.
We run a mirrored~\cite{mirror} GHZ preparation circuit comprised of transversal gates in this code and observe that, as circuit depth increases, the LQEC scheme becomes dominated by errors from noisy inter-processor gates, while the DQEC scheme remains comparatively robust.

One apparent obstacle to realizing fully local gates through DQEC is that non-transversal gates will require substantial overhead; e.g., distributed magic state distillation.
Indeed, the above simulation only involves gates which are transversal in the $[[7,1,3]]$ Steane code (H and CNOT), and the observed advantage therefore does not extend to general circuits for this code.
We explore a method to circumvent this limitation using approximate codes in Section~\ref{adqec}.

\subsubsection{Resilience to Spatially Correlated Errors} \label{exact-distributed-spatially-correlated-errors}

\begin{figure} 
\includegraphics[width=\columnwidth]{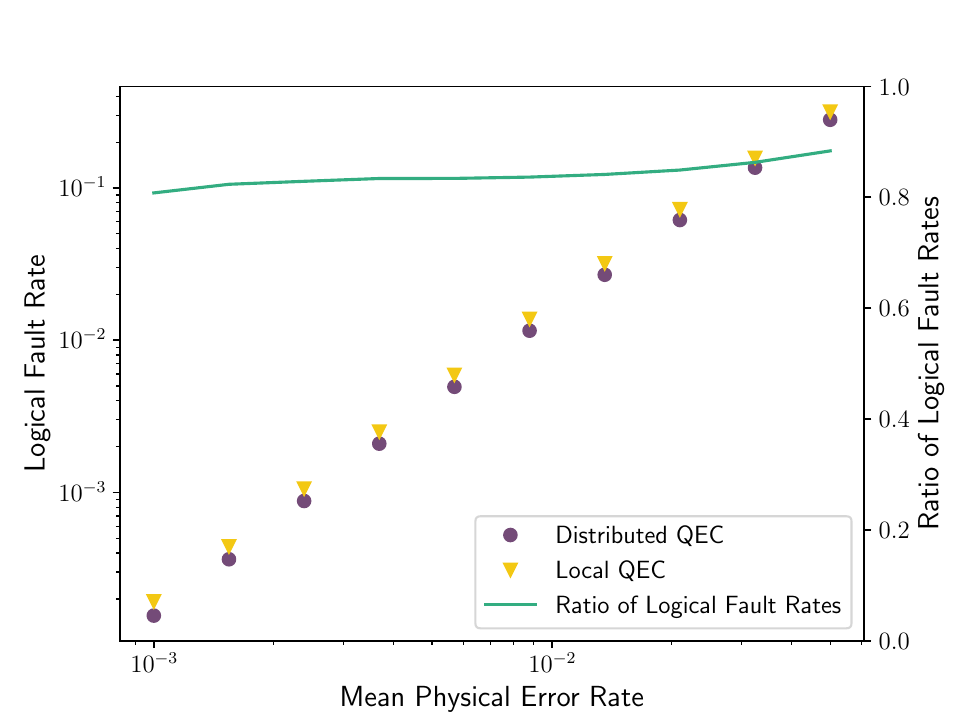}
\centering
\caption{Logical error rates of distributed QEC (red dots) vs local QEC (blue dots) under spatially correlated error conditions. We simulate seven $[[7,1,3]]$ Steane codes across seven processors, comparing the cases where code blocks are each contained within their own processor (local) and each spread among all processors (fully distributed). The processor error rates are sampled from normal distributions, with the standard deviation set to 0.5 times the mean error rate. Distributed code blocks consistently outperform local code blocks, resulting in 13-20\% lower logical error rates (green curve).}
\label{dqec-vs-lqec}
\end{figure}

Another key advantage of DQEC in distributed systems is its capacity to manage non-uniform error environments more gracefully than local QEC. In a large-scale distributed quantum computer, hardware inhomogeneity is inevitable as some modules may host qubits with longer coherence times than others, certain communication channels might suffer elevated loss or delay, and devices may drift out of calibration at different rates. This challenge is amplified in heterogeneous architectures where processors have varying qubit connectivities and may even be based on different qubit technologies. Additionally, error processes within a module, such as qubit crosstalk or burst errors~\cite{jiang-cosmic-ray-errors}, tend to be spatially localized and correlated. 
Crosstalk, in particular, is a prevalent source of error which can induce multiple errors within a single code block when code blocks are composed of neighboring qubits.

The result is that the distribution of errors across a realistic modular system is not i.i.d.\ -- in particular, errors are more likely to be spatially clustered. Consequently, distributing code blocks in space spreads these local errors across different code blocks, lowering the probability that the number of errors rises above the quantity that can be corrected by the code. 

In this sense, DQEC strictly outperforms LQEC under spatially correlated errors.
We demonstrate this result through simulation in Figure~\ref{dqec-vs-lqec} for $[[7,1,3]]$ Steane codes in a seven-processor device with uneven error rates across modules. This confirms our intuition that spreading errors evenly among code blocks results in a lower overall probability of fault.

\subsection{Limitations of Existing Work on DQEC}  \label{existing-work}

Despite the potential for advantage from distributed error correction in the multi-QPU setting, standard approaches involving exact codes are not well suited to exploit this potential. As a consequence of the Eastin-Knill theorem, exact error-correcting codes cannot implement a universal set of transversal gates, and any DQEC scheme based on such codes therefore must incur additional overhead or complexity.

The conventional method to circumvent this issue is to apply magic state injection to realize the non-transversal gate. The distillation of magic states is notoriously resource-intensive: even with optimized protocols, distilling a single $|T\rangle$ state with infidelity $10^{-12}$ from states of infidelity $10^{-4}$ incurs a qubit overhead of \mytilde$10^4$ and a gate overhead of at least this amount~\cite{cost-of-universality}. Existing distributed QEC schemes~\cite{buffalo-distributed, cisco-distributed-color-code, nu-quantum-floquet} necessitate \emph{distributed} magic state distillation, incurring a potentially overwhelming overhead.

Another workaround to the Eastin-Knill theorem is to employ code switching, a technique which alternates encoding between two different QEC codes, which together are capable of implementing a universal set of transversal gates. While conceptually feasible, code switching requires fault-tolerantly mapping the entire encoded state between distinct codes, which is a complex process involving many extra gates and measurements. A recent study~\cite{cost-of-universality} directly comparing magic state distillation and code switching for 2D color codes (often cited as one of the most promising QEC schemes) found that code switching is more resource-intensive than state distillation, with a T-gate threshold nearly an order of magnitude lower. Thus, both approaches carry heavy costs for large circuits.

Xu et al.~\cite{jiang-cosmic-ray-errors} introduce a DQEC scheme for chip-level erasure errors which exemplifies how spatially correlated errors may be mitigated by distributed error correction. Their work focuses on a specific source of noise, cosmic ray events, which are relatively rare compared to computational errors.
In our paper, we offer methods which extend the thesis of~\cite{jiang-cosmic-ray-errors} to more general sources of error through approximate codes.

Finally, while previous works have focused on the improved error correction performance derived from distribution, they do not quantify the reduction in processor-nonlocal gates that results from distributing logical qubits in such a way that logical gates can be performed processor-locally. We give the first results regarding the minimization of these noisy gates.

\section{Approximate Error Correction via the W-State Code} \label{sec:w-state-code}

Our proposed approach centers around distributing \emph{approximate} quantum error-correcting codes in a multi-processor setting. Approximate codes sacrifice the deterministic error correcting ability of traditional codes in favor of additional desirable properties, such as universal transversal gate sets. Further advantages of approximate codes arise in the distributed setting and will be discussed in Section~\ref{adqec}. 
Here we lay out the foundations of the permutation-invariant W-state code, a simple but representative approximate quantum error-correcting code. To the best of our knowledge, the preparation circuits presented here are the first explicit constructions of the W-state code, and we believe them to be of independent interest to the distributed aspects of this paper as their methods may be generalizable to  encoding and decoding schemes for other permutation-invariant codes.

\subsection{W State Preparation}

The W-state code is an approximate quantum error-correcting code whose encoding resembles the W state~\cite{eczoo-w-state}. 
%As we will see, this code evades the Eastin-Knill theorem and admits several advantageous properties in the distributed setting.
An $n$-qubit W state is the equal superposition of all basis states with a single excitation:
\begin{equation}
|W_n\rangle = \frac{1}{\sqrt{n}}\left(|100\dots0\rangle + |010\dots0\rangle + \cdots + |0\dots01\rangle \right)
\end{equation}

The W state can be generalized to systems of higher dimensions. For a $d$-dimensional system, the W state is often defined as:
\begin{equation}
\begin{split}
|W_n\rangle = \frac{1}{\sqrt{(d-1)n}}\sum_{j=1}^{d-1}\big(|j00\dots0\rangle + |0j0\dots0\rangle \\ + \cdots + |0\dots0j\rangle \big) \label{eq:quditW}
\end{split}
\end{equation}
Yeh~\cite{yeh-qudit-w-state} gives constructions for deterministic preparation of certain W states through the introduction of a novel non-Clifford qudit gate (the $\sqrt[d]{Z}$ gate). We improve on her construction, providing a circuit which scales a W state with $n=2^k$ qudits ($k\in \mathbb{Z}^+$) into a W state with $n'=2^{k+1}$ qudits by adapting a well-known construction to the qudit case.
Our construction, shown in Figure~\ref{fig:w-prep}, requires only a qubit (not qudit) ancilla and only uses $O(dn)$ non-Clifford gates (as opposed to $O(n^2)$ in Yeh's construction), but only works for W states whose size is a power of 2. In the context of a many-processor quantum computer, however, this is not necessarily a significant limitation. The circuit is recursive starting with an easy-to-prepare 2-qudit W state. In the qubit case, this is simply the Bell state $\ket{\Psi^+}=\frac{1}{\sqrt2}\left(\ket{01}+\ket{10}\right)$; in the qudit case this state is given by Equation~\ref{eq:quditW}.% and the construction of this state is detailed in Figure~\ref{fig:w-prep-2}.   

\begin{figure}
\centering
    \begin{subfigure}{\columnwidth}
        \includegraphics[width=.5\textwidth]{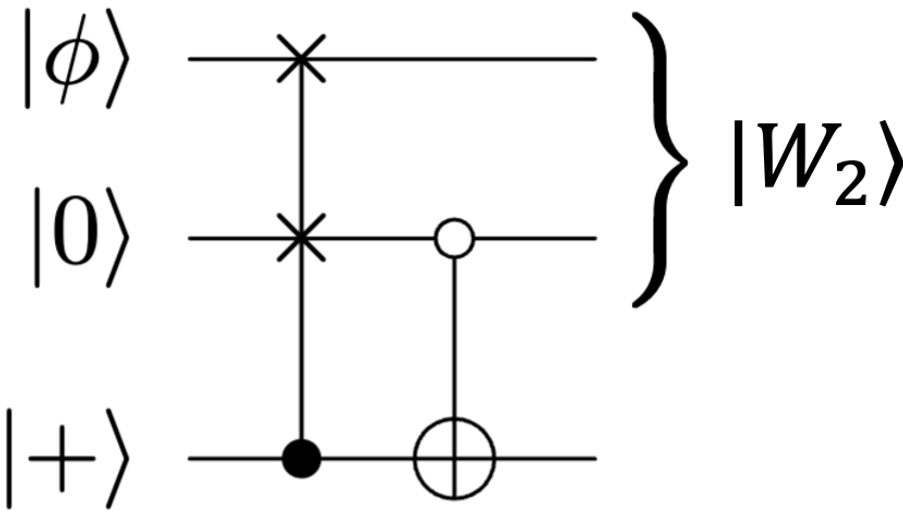}
        \caption{Preparation circuit for a qudit W state of size $n=2$ and dimension $d$, where $\ket{\phi}=\frac{1}{\sqrt{d-1}}\sum_{i=1}^{d-1}\ket{i}$.}
        \label{fig:w-prep-2}
    \end{subfigure}    
    \begin{subfigure}{\columnwidth}
        \includegraphics[width=\textwidth]{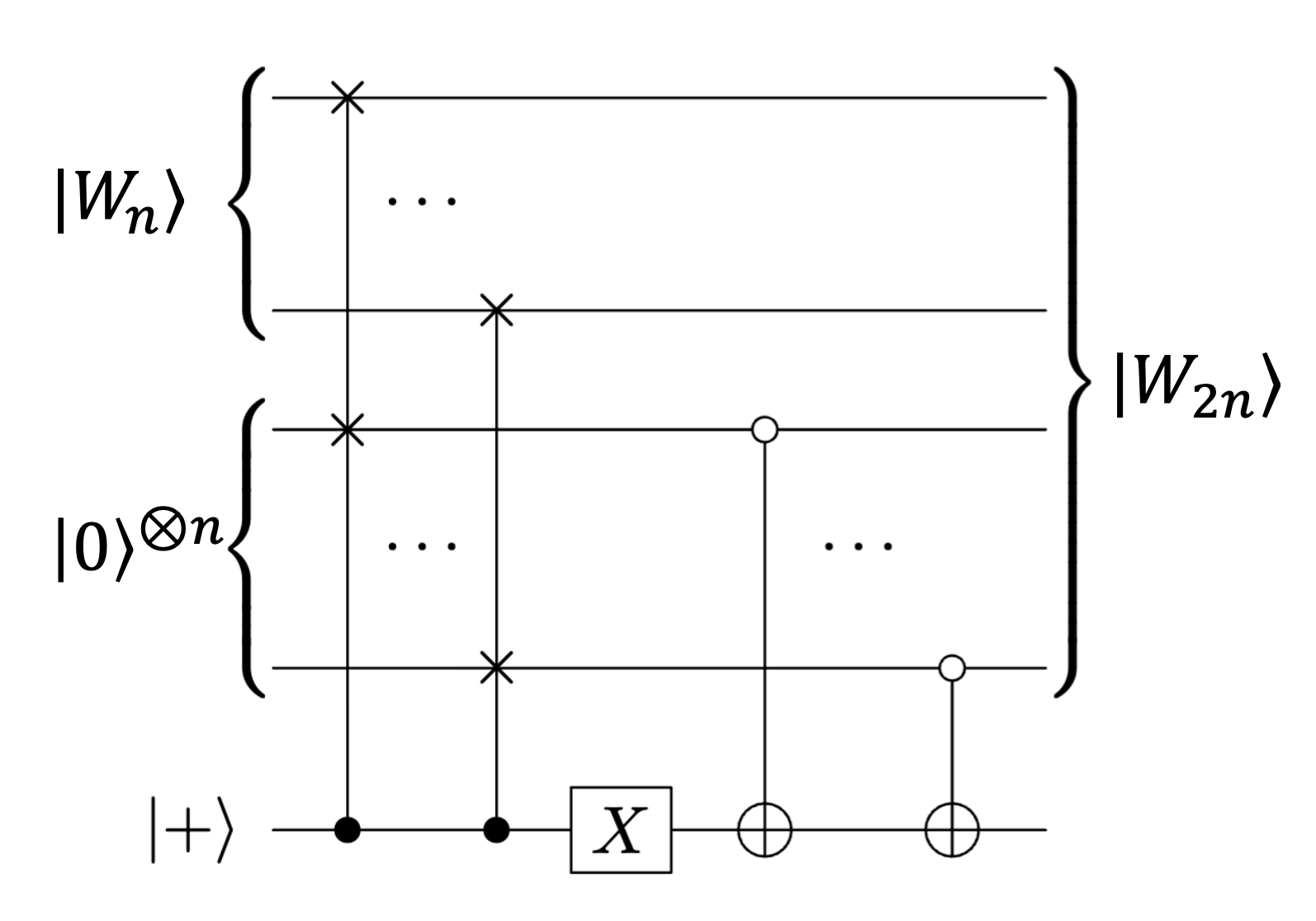}
        \caption{The W state scaling circuit can double the size of any qudit W state.}
    \end{subfigure}
    \caption{W state preparation procedure. We construct a $2^k$-qudit W state by first (a) constructing the 2-qudit W state and then (b) recursively doubling its size. In the $d=2$ case we can replace the circuit in (a) with a simple Bell state preparation circuit, which is composed only of Clifford gates.
Note that the ancilla ($\ket{+}$) in each circuit is a qubit, while other subsystems are qudits. The open-circle CNOT gates are controlled-on-$\ket{0}$-NOT gates; a decomposition is given in~\cite{yeh-qudit-w-state}.
}
    \label{fig:w-prep}
\end{figure}

\subsection{The W-State Code} \label{w-state-code}

\subsubsection{Definition}
The W-State code encodes a $d_L$-dimensional logical system into $n$ physical qudits of dimension $d_L+1$ in a state analogous to an $n$-partite $W$ state. The encoding of a logical state $|\psi\rangle$ is
\begin{equation}
\ket{\psi}_L = \frac{1}{\sqrt{n}}\left(|\psi\bot\bot\dots\bot\rangle + |\bot\psi\bot\dots\bot\rangle + \cdots + |\bot\dots\bot\psi\rangle \right)
\end{equation}
where $|\bot\rangle$ indicates the $(d_L+1)$-th basis state.
Then for $d_L=2$ we have the following logical states:
\begin{align}
\begin{split}
\ket{0}_L &= \frac{1}{\sqrt{n}}\left(\ket{022\dots2} + \ket{202\dots2} + \cdots + \ket{2\dots20} \right) \\ 
\ket{1}_L &= \frac{1}{\sqrt{n}}\left(\ket{122\dots2} + \ket{212\dots2} + \cdots + \ket{2\dots21} \right)
\end{split}
\end{align}

The intuition for the W-state code is that it encodes the location of the logical state in a superposition of the $n$ physical systems.
If any one of the $n$ qudits is lost or corrupted, the logical state is not immediately collapsed --- it is only affected with probability $1/n$.

\subsubsection{Transversal Gates in the W-State Code}

The W-state code readily admits a universal set of transversal gates: to apply a logical unitary $U$, one can simply apply $U$ to each $d$-dimensional physical qudit in a way that acts trivially on $|\bot\rangle$, i.e., $U\ket{\bot}=\ket{\bot}$~\cite{eczoo-w-state}.
This mechanism differs from that of most block codes.

This works directly for single-qudit unitaries since the W-State code is $U(d)$-covariant. 
A single-qudit logical unitary $U_L$ can be performed as $U_L=U^{\otimes n}$:
\begin{equation}
    U_L\ket{\psi_L}=\frac{1}{\sqrt{3}}\Big((U\ket{\psi})\ket{\bot\bot}+\ket{\bot}(U\ket{\psi})\ket{\bot}+\ket{\bot\bot}(U\ket{\psi})\Big)
\end{equation}

Notably, this approach does not work for logical 2-qudit unitaries between two separately-encoded W-state code blocks.
Instead, we require a construction that is $U(d^2)$-covariant, which we accomplish by encoding both relevant qudits into a single code block.
To encode $k$ d-dimensional qudits $\ket{\psi}:=\ket{\psi_1\cdots\psi_k}$, we simply set $d_L=d^k$ and define $\ket{\bot}:=\ket{d}^{\otimes k}$. 
When encoding qubits, this construction allows us to still use qutrits for arbitrarily large $d_L$.
As an explicit example, consider encoding two qubits $\ket{\psi}$ and $\ket{\phi}$ into four physical subsystems: we have the logical state
\begin{equation}
    \ket{\psi}_L= \frac{1}{\sqrt2}\Big(\ket{\psi\phi 22}+\ket{22\psi\phi}\Big),
\end{equation} 
Under this construction, single-qubit gates can still be performed transversally on the desired qubits and logical CNOTs can now be performed transversally as $CNOT_{1,2}\otimes CNOT_{3,4}$. 
%However, since these gates are no longer code-block transversal but rather transversal within the code, it is prudent to perform flag checks for leakage into the qutrit subspace after transversal two-qubit gate implementation. 

The W-state code achieves this trivial sidestepping of the Eastin-Knill theorem at the cost of requiring higher dimensionality and non-deterministic decoding.
While there exist other permutation-invariant codes with larger, non-trivial distance that can support some logical transversal gates~\cite{kubischta-PI-transversal-phase, ouyang-measurement-free-code-switching-PI}, these require code switching for universal transversality.

\subsubsection{Encoding Into the W-State Code}

\begin{figure}
    \centering
    \begin{subfigure}{\columnwidth}
        \includegraphics[width=\textwidth]{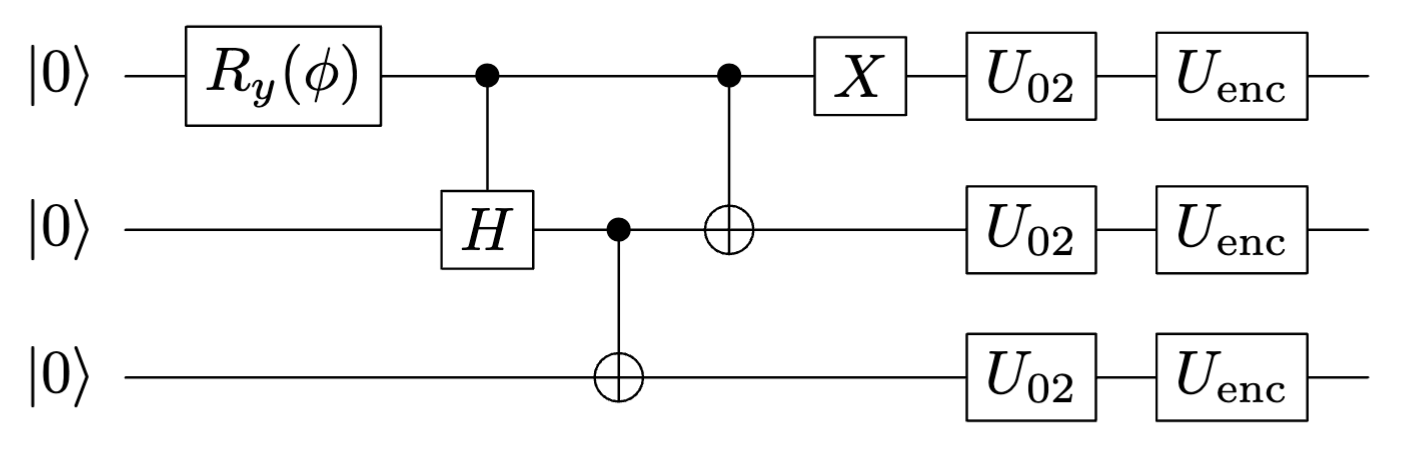}
        \caption{Encoding into the $n=3$ W-state code, requiring analog rotations for encoding and state preparation.}
        \label{fig:w-encode-3}
        
    \end{subfigure}
    \begin{subfigure}{\columnwidth}
        \includegraphics[width=.6\textwidth]{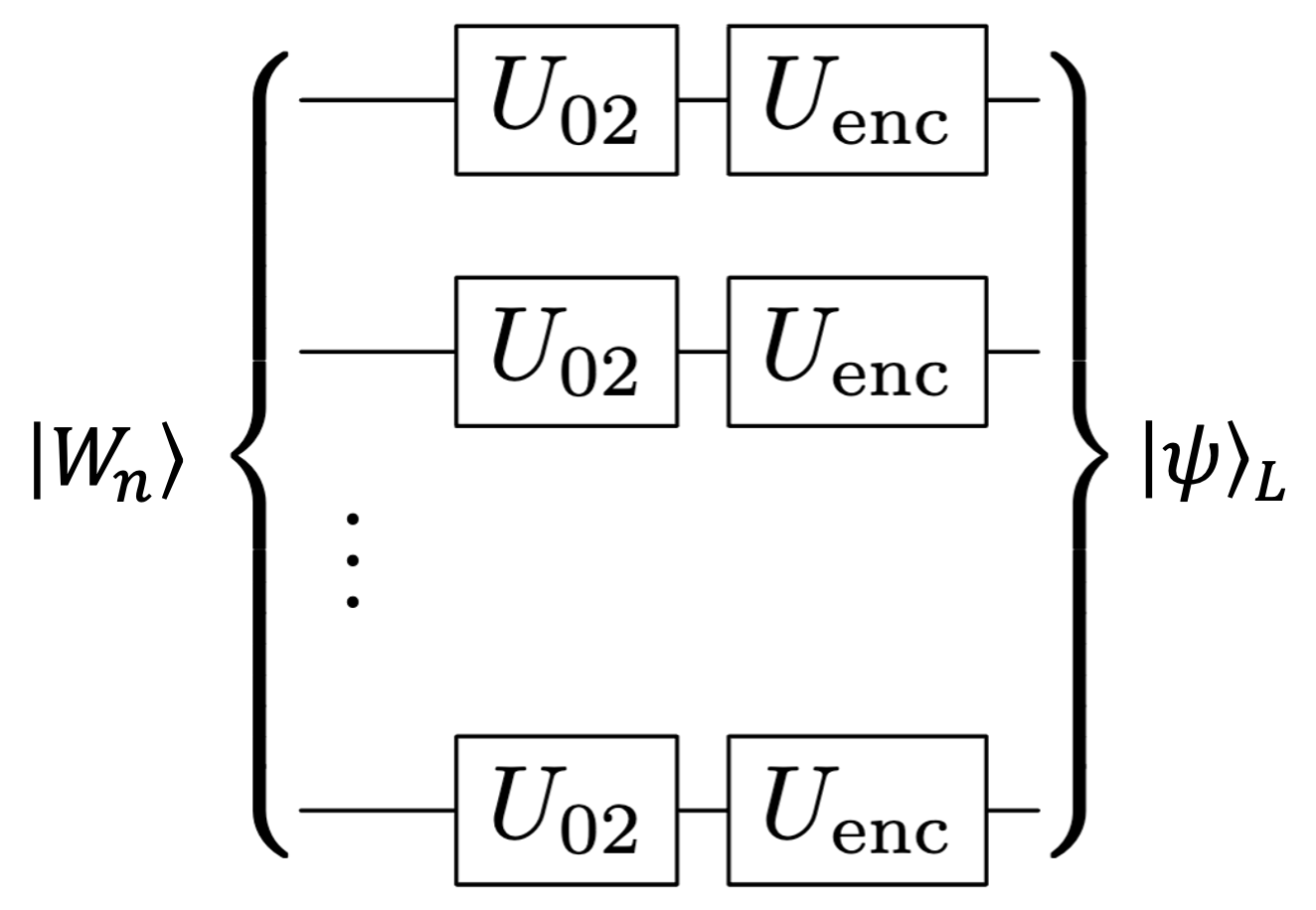}
        \caption{General W-state encoding using the analog operators $U_{02}=\ket{0}\bra{2}+\ket{2}\bra{0}+\ket{1}\bra{1}$ and $U_{\text{enc}}=\ket{\psi}\bra{1}+\ket{\psi_\perp}\bra{0}+\ket{2}\bra{2}$.}
    \end{subfigure}
    \caption{Encoding circuit for the W-state code. The encoding procedure consists of a W state preparation followed by analog rotations. (a) Encoding circuit for $n=3, d_L=2$, involving a Y-rotation about $\phi=2\cos^{-1}(1/\sqrt{3})$. (b) General encoding circuit. The encoding simplifies significantly for W states with $n=2^k$, as the input W state can be prepared according to Figure~\ref{fig:w-prep}. }
    \label{fig:encoder}
\end{figure}

Although the W-state code does not directly use the canonical W state as a codeword, the preparation of the code state follows a process similar to that of constructing a W state (see Figure~\ref{fig:encoder}). 
We make this explicit for the $n=3, d_L=2$ case and illustrate this case in Figure~\ref{fig:w-encode-3}. First, one prepares the W state $\ket{W_3}=\ket{100}+\ket{010}+\ket{001}$. Next, one applies the operator $U_{02}=\ket{0}\bra{2}+\ket{2}\bra{0}+\ket{1}\bra{1}$ to the state $\ket{W}$ as $U_{02}^{\otimes n}$. This encodes the logical $\ket{1}$ state 
\begin{equation}
    \ket{1}_L=\frac{1}{\sqrt3}(\ket{122}+\ket{212}+\ket{221}).
\end{equation}
We then define an encoding unitary $U_{\text{enc}}=\ket{\psi}\bra{1}+\ket{\psi_\perp}\bra{0}+\ket{2}\bra{2}$, where $\ket{\psi}=c_0 \ket{0}+c_1\ket{1}$ is the qubit state we wish to encode and $\ket{\psi_\perp}$ is an orthogonal state $\ket{\psi}$ in the $\{\ket{0},\ket{1}\}$ qubit subspace (required for unitarity). To encode, we apply the operator $U_\text{enc}^{\otimes n}$, encoding into each subsystem in turn and producing the logical encoded state
\begin{equation} \label{eqn:logical-w}
\ket{\psi}_L=\frac{1}{\sqrt3}(\ket{\psi 22}+\ket{2\psi 2}+\ket{22\psi}).
\end{equation}  

This procedure easily generalizes to the encoding of more than one logical qubit, as shown in Figure~\ref{fig:2qtransversal} for the case of encoding $k=2$ logical qubits into four physical subsystems.

\begin{figure}
    \centering
    \includegraphics[width=.6\columnwidth]{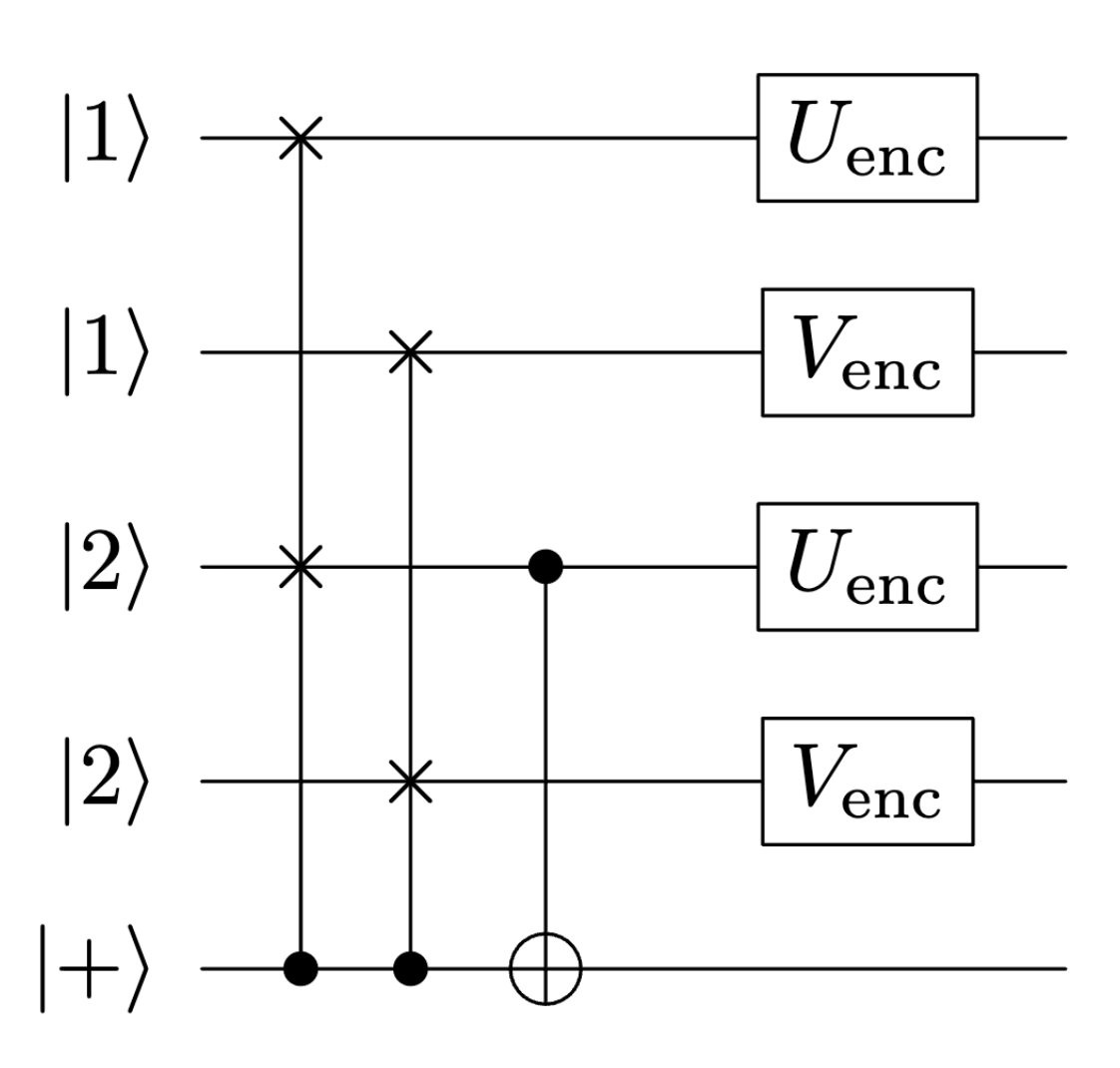}
    \caption{The single-qubit states $\ket{\psi}$ and $\ket{\phi}$ can be encoded into a $k=2$ W-state code using unitaries $U_{\text{enc}}$ (previously defined) and $V_{\text{enc}}=\ket{\phi}\bra{1}+\ket{\phi_\perp}\bra{0}+\ket{2}\bra{2}$.}
    \label{fig:2qtransversal}
\end{figure}

\subsubsection{Decoding the W-State Code} \label{w-state-decoding}

We consider decoding under erasure errors, as these are an increasingly dominant noise channel \cite{quantum-circuit-inc, leakage-to-erasure-conversion-aws, neutral-atom-erasure} and are the standard error channel analyzed in connection with approximate error correcting codes. An erasure error on the first subsystem of the state in Equation~(\ref{eqn:logical-w}) transforms the state as 
\begin{equation}
    \rho_{\varnothing}=\ket{\varnothing}\bra{\varnothing}\otimes\rho_e,
\end{equation}
where
\begin{equation}
\rho_e := \frac13\ket{22}\bra{22}+\frac23\left(\ket{\psi 2}+\ket{2\psi}\right)(\bra{\psi 2}+\bra{2\psi})
\end{equation}
and $\ket{\varnothing}$ denotes the erased subsystem.
This leaves us to decode $\rho_e$, which resembles a Bell state plus a noise term. We wish to apply the projector $\ket{12}(\bra{21}+\bra{12})+\ket{02}(\bra{20}+\bra{02})+\ket{22}\bra{22}$.
For the case of two unerased subsystems and $d_L=2$, this can be realized by the circuit in Figure~\ref{fig:decoder-2-with-measure}. 
The procedure for decoding in the case of erasures on the second or third subsystem is symmetric.
%A symmetric approach can be taken to correct errors on the second or third subsystem of the code, as flagged by location of the erasure. If there are no erasures, the circuit in Figure~\ref{fig:noedecode} should be applied instead to restore the original state $\ket{\psi}$. 

\begin{figure}
    \centering
    \begin{subfigure}{\columnwidth}
    \includegraphics[width=\columnwidth]{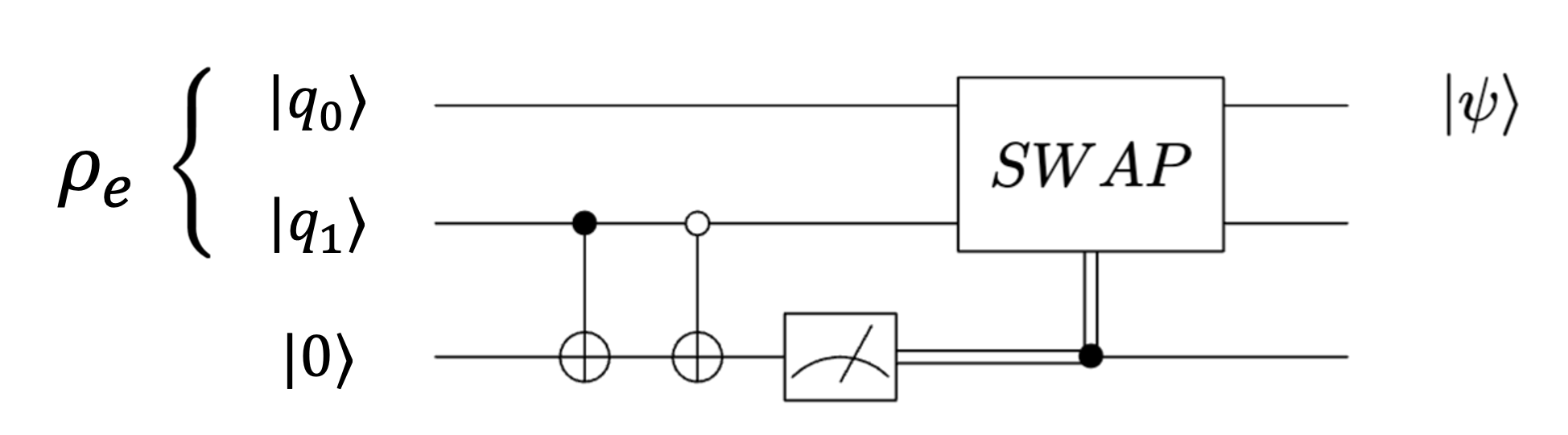}
    \caption{W-state decoder for $n=2$. The measurement-conditioned SWAP ensures that the decoded $\ket{\psi}$ is located in the first subsystem.}
    \label{fig:decoder-2-with-measure}
    \end{subfigure}
    \begin{subfigure}{\columnwidth}
    \includegraphics[width=\columnwidth]{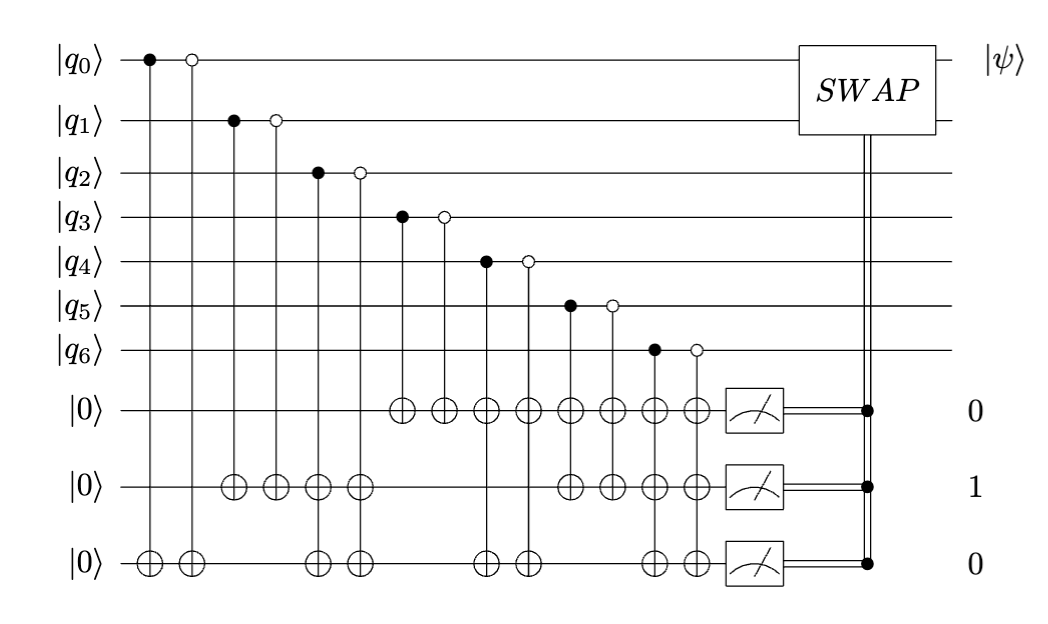}
    \caption{W-state decoder for $n=7$. The measurement outcomes indicate the decoded location of $\ket{\psi}$ (in this example \texttt{010} indicates $\ket{\psi}$ is in $\ket{q_1}$), which is then swapped into the first subsystem as desired. Measuring \texttt{000} indicates decoding failure.}
    \label{fig:decoder-7-with-measure}
    \end{subfigure}
    \caption{Decoding circuits for the W-state code involving mid-circuit measurements. The input to the decoder is an encoded state with $d_L=2$ and $n$ unerased subsystems. All ancillas are qubits. Figure (a) shows a minimal example with $n=2$ while (b) shows a larger instance with $n=7$. These circuits act trivially over erased subsystems, causing a failure with probability $\frac{n_e}{n+n_e}$ for a system with $n_e$ erasures. These circuits easily generalize to any $n$, requiring $\lceil \log_2(n+1) \rceil$ ancilla qubits and $O(n\log n)$ CNOT gates. Note that the ancillas are qubits while the other subsystems are qutrits.}
    \label{fig:decoder}
\end{figure}

In Figure~\ref{fig:decoder} and subsequent figures, the closed-circle CNOT represents a controlled-on-$\ket{1}$-CNOT where the control subsystem is a qudit and the target is a qubit. Similarly, the open-circle CNOT represents a controlled-on-$\ket{0}$-CNOT.
For the W-state codes considered in this paper with $d_L=2$ (and therefore subsystem dimension $d=d_L+1=3$), the action of a pair of these gates is to indicate in the ancilla the presence or absence of $\psi$. For example, defining $V$ to be the first two CNOTs in Figure~\ref{fig:decoder-2-with-measure}, the result of these gates is
\begin{equation}
\begin{split}
V\left(\rho_e\otimes\ket{0}\bra{0}\right) V^\dagger = \frac13\ket{22}\bra{22}\otimes\ket{0}\bra{0} \hspace{2.5cm} \\
+\frac23\left(\ket{\psi 2}\ket{0}+\ket{2\psi}\ket{1}\right)(\bra{\psi 2}\bra{0}+\bra{2\psi}\bra{1})
\end{split}
\end{equation}
so that the ancilla is 1 when $\psi$ is in the second position.
%The ancilla can then be used to indicate the location of $\psi$.

In this example with $n=2$ and one erasure, the decoding fails with probability $1/3$. However, as the size $n$ of the code increases, this failure occurs only with probability $1/n$, approaching the performance of an exact error-correcting code in the asymptotic limit.
%As shown in Figure~\ref{fig:decoder}, this decoding can be performed with or without a mid-circuit measurement, the latter of which eliminates the need for costly CSWAP gates and reduces the gate count by a Hadamard. 
It is also worth noting that in this measurement decoding case, the SWAP gate need only be applied half of the time. 

The decoding procedure in Figure~\ref{fig:decoder-2-with-measure} scales to $n$ (unerased) subsystems through a simple generalization. Figure~\ref{fig:decoder-7-with-measure} shows an example with $n=7$.
For each subsystem $q_i$, a pair of CNOTs encodes the binary representation of $i+1$ into $\lceil \log_2(n+1) \rceil$ ancillas.
The result of measuring the ancillas then indicates the decoded location of $\ket{\psi}$, while a measurement of all 0's indicates a decoding failure.

\subsection{Elective Decoding}

The W-state code exhibits a remarkable property: the state $\ket{\psi}$ may be decoded into any of the (non-erased) constituent physical qubits using fewer than one swap in expectation. We call this property \emph{elective decoding} to emphasize our ability to choose the resultant location of the decoded system.

To understand elective decoding, let us first consider the 3-qubit repetition code as an example. Figure~\ref{fig:repcode} shows how it is possible to decode into any of the 3 qubits comprising this code without additional swap gates. 
However, the 9-qubit Shor code (Figure~\ref{fig:shorcode}), which is built on repetition codes, is only able to decode into one of the three qubits which comprise the outer repetition code, since the encoding Hadamards were already positioned ahead of time. The qubits that executed a Hadamard during encoding must also execute a Hadamard in decoding, limiting our decoding options to those three locations.
In these examples, we assume in-place decoding, which requires non-Clifford gates and is generally not favored over syndrome extraction; however, the elective-decodability of the repetition code may itself be of interest, as noise-biased qubits such as \cite{catqubits} sometimes make use of the repetition code.

\begin{figure}
\centering
\begin{subfigure}{\columnwidth}
    \includegraphics[width=\columnwidth]{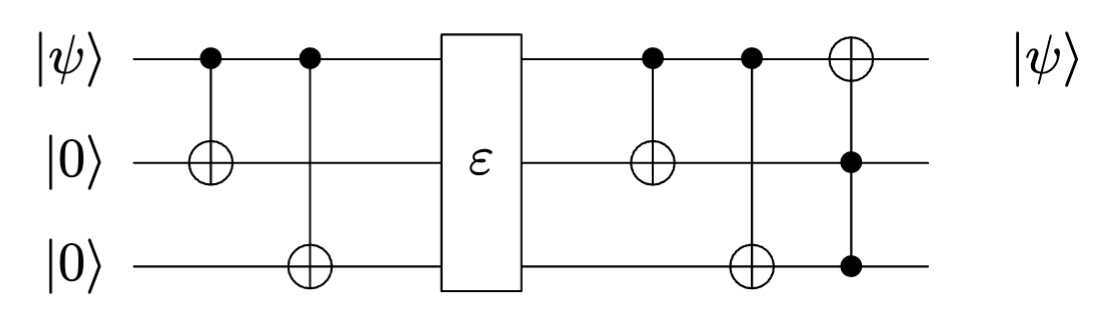}
    \includegraphics[width=\columnwidth]{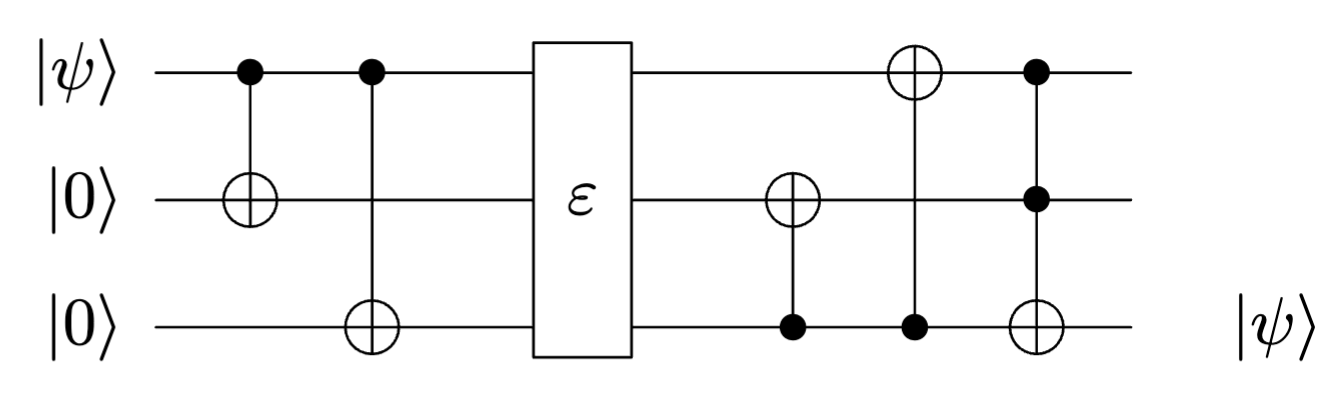}
    \caption{Two decoding circuits for the 3-qubit repetition code under one X-type error ($\varepsilon$). This code exhibits elective decoding, as it is possible to decode into a different qubit without additional swap gates.}
    \label{fig:repcode}
\end{subfigure}
\begin{subfigure}{\columnwidth}
    \includegraphics[width=\columnwidth]{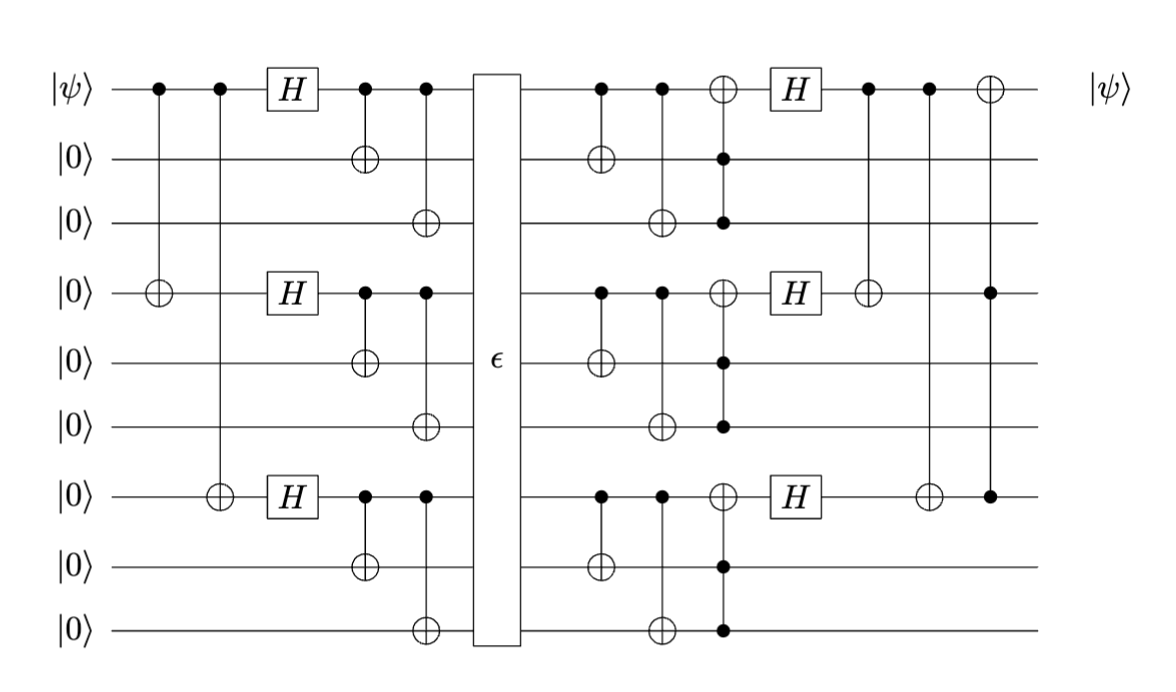}
    \caption{The Shor code does not exhibit elective decoding. One may only decode into one of the qubits which initially received a Hadamard.}
    \label{fig:shorcode}
\end{subfigure}
\caption{Elective decoding: (a) example and (b) non-example. The W-state decoders in Figure~\ref{fig:elective-decoding} also exhibit elective decoding.}
\end{figure}

We provide an alternate W-state code decoder which allows for elective decoding. %, shown in Figure~\ref{fig:elective-decoding}.
Figure~\ref{fig:decoder-2-elective} shows how this can be achieved for the $n=2$ case by replacing the measurement and SWAP in Figure~\ref{fig:decoder-2-with-measure} with a conditional SWAP.
We extend this approach to general $n$ by applying CNOTs and CSWAPs in a way which systematically moves all $\psi$ terms to the desired decoding location.
The algorithm recursively chooses half of the possible locations for $\psi$ terms and swaps any terms in these locations to one of the other locations by conditioning on a fresh ancilla.
This cuts the number of locations containing $\psi$ terms in half each round, which results in $\lceil \log_2 n \rceil$ rounds and therefore the same number of ancillas, as well as $n-1$ CSWAPS and $O(n)$ CNOTs.
We give the explicit circuit for this decoder in the $n=7$ case in Figure~\ref{fig:decoder-7-with-measure}, where the desired decoding location is $\ket{q_6}$.
After decoding, the ancillas are unentangled with the rest of the system. To return them to $\ket{0}^{\otimes\lceil \log_2 n \rceil}$, we can apply $H^{\otimes\log_2 n}$ when $n$ is a power of two as in Figure~\ref{fig:decoder-2-elective}; in other cases, we require an explicit reset.

\begin{figure}
\centering
    \begin{subfigure}{\columnwidth}
    \includegraphics[width=.8\columnwidth]{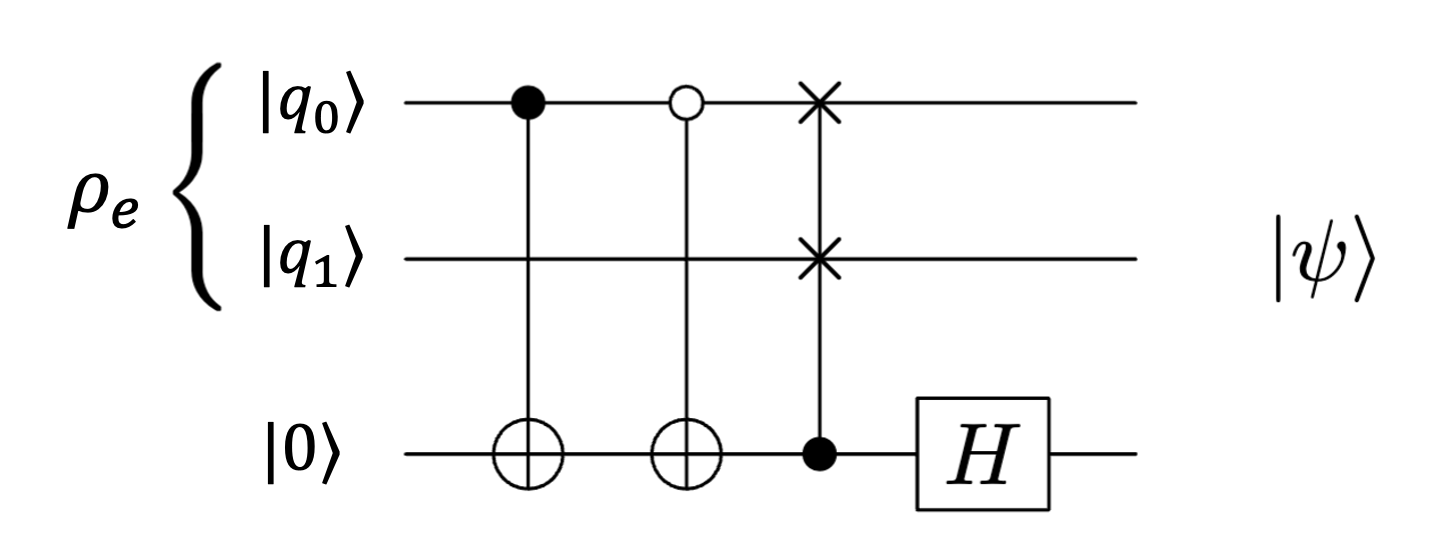}
    \caption{W-state decoder allowing elective decoding without measurement for $n=2$.}
    \label{fig:decoder-2-elective}
    \end{subfigure}
    \begin{subfigure}{\columnwidth}
    \includegraphics[width=\columnwidth]{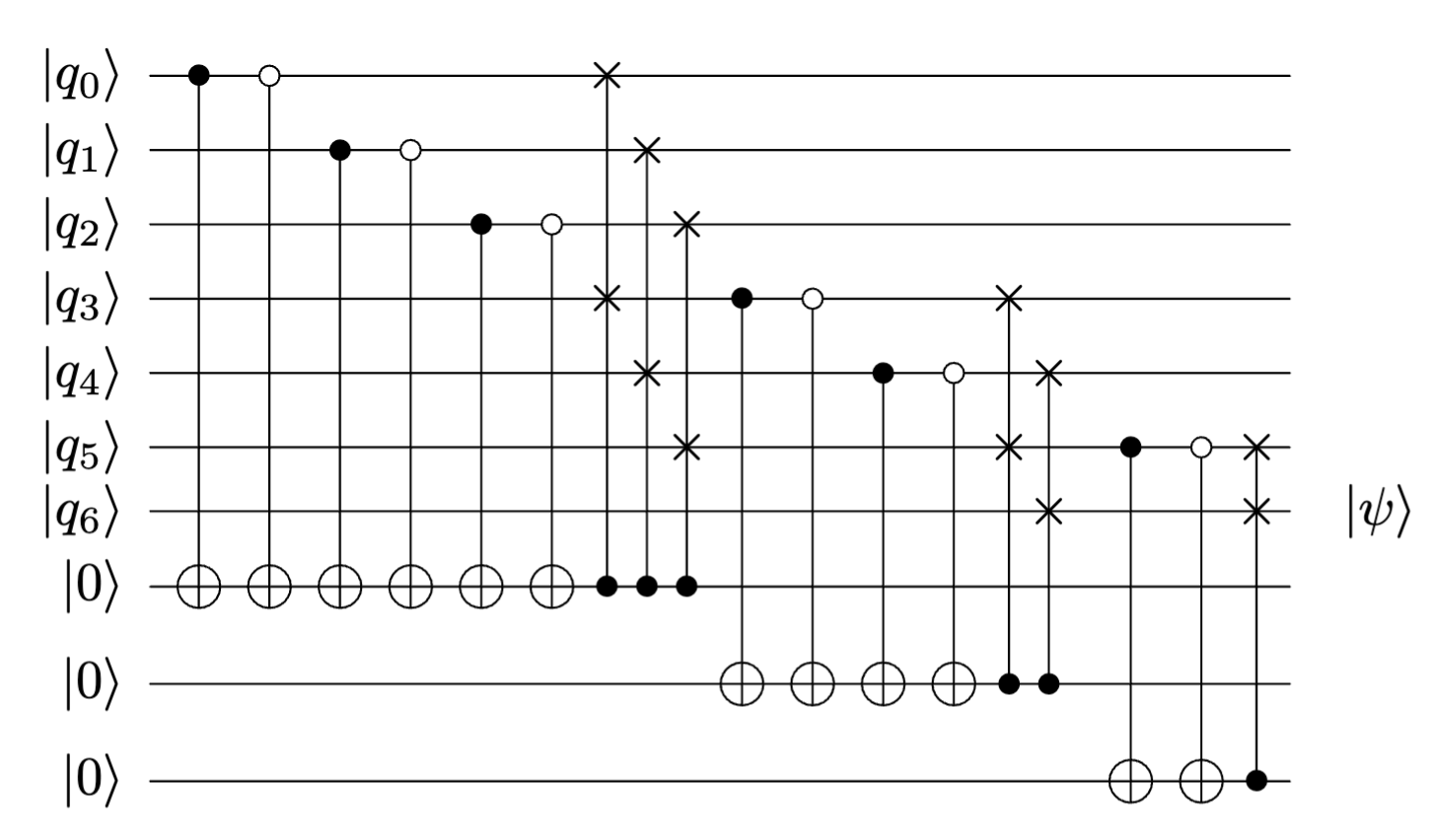}
    \caption{W-state decoder allowing elective decoding for $n=7$. In this instance, the desired decoding location is $\ket{q_6}$.}
    \label{fig:decoder-7-elective}
    \end{subfigure}
    \caption{Alternate decoding circuits for the W-state code exhibiting elective decoding. These allow the choice of decoding location to be deferred until decoding is performed. The decoders operate by recursively moving all $\psi$ terms to the desired subsystem, using $\lceil \log_2 n \rceil$ ancillas, where $n$ is the number of unerased subsystems. Note that the ancillas are qubits while the other subsystems are qutrits.}
    \label{fig:elective-decoding}
\end{figure} 

%In a distributed context, this can have a high impact in terms of reducing the number of swap gates required by the overall circuit, as mid-circuit decoding and encoding may be necessary for high-depth circuits. With Agnostic Decoding, it is no longer necessary to teleport qubits to recreate codes, if the code can simply be decoded into that qubit.

\subsection{Fault-Tolerant Encoding and Decoding}

\begin{figure}
\centering
    \begin{subfigure}{\columnwidth}
        \includegraphics[width=.6\textwidth]{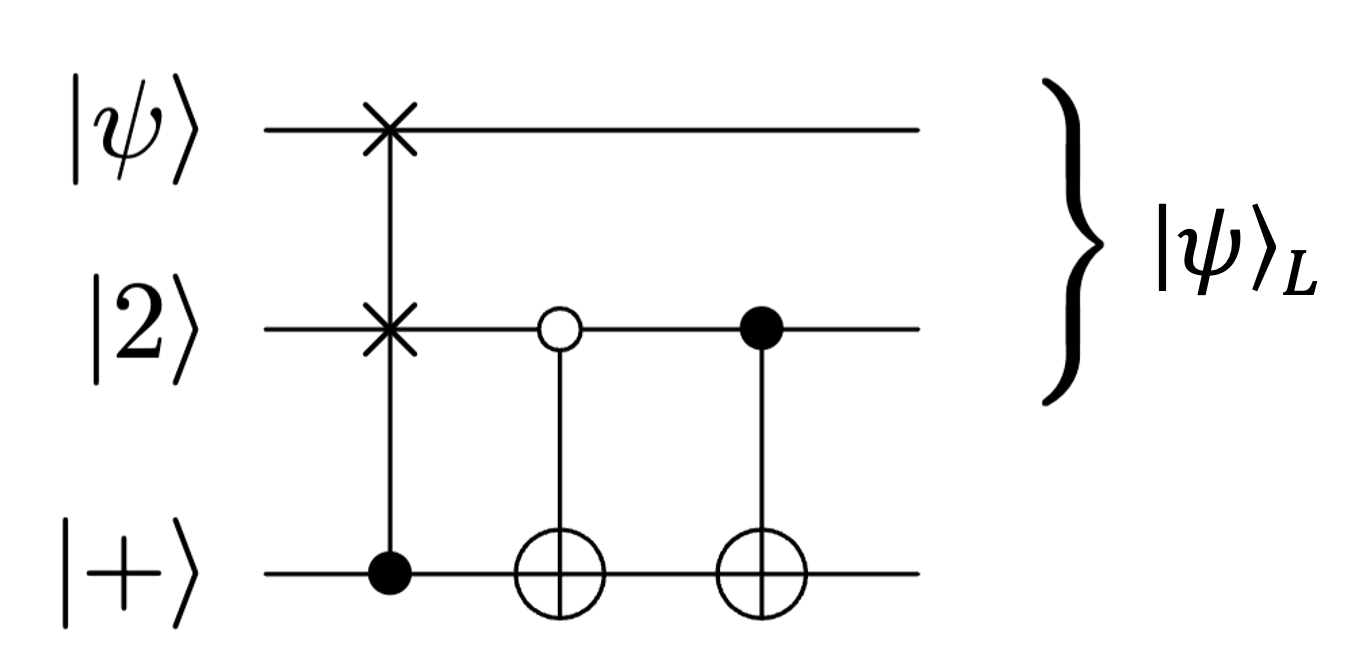}
        \caption{Minimal instance of alternative W-state encoder, encoding into a logical state of size 2.}
        \label{fig:w-prep-2}
    \end{subfigure}    
    \begin{subfigure}{\columnwidth}
        \includegraphics[width=\textwidth]{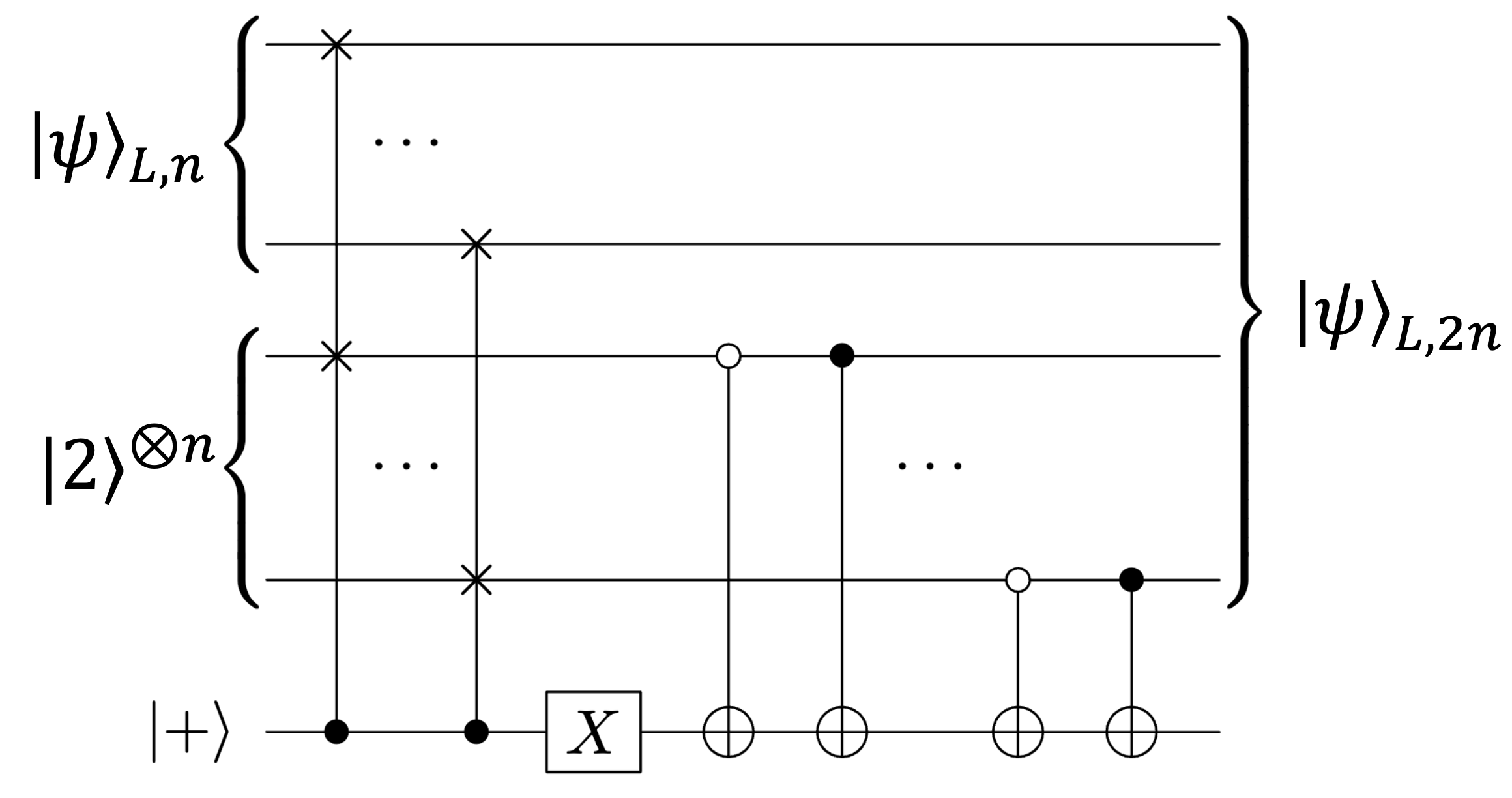}
        \caption{Scaling circuit doubles the size of a W-state code.}
    \end{subfigure}
    \begin{subfigure}{\columnwidth}
        \includegraphics[width=.6\textwidth]{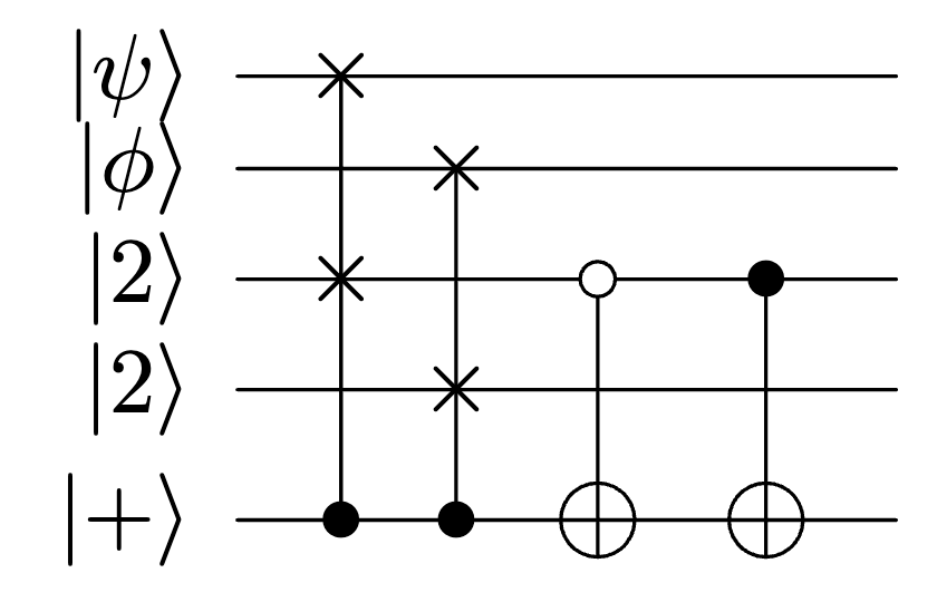}
        \caption{Minimal instance of alternative W-state encoder for encoding of 2 logical qubits. Fewer ancilla checks are required than for a code of the same size that encodes only 1 qubit.}
    \end{subfigure}
    \caption{An alternative encoding circuit for the W-state code following a construction similar to that of the W state preparation circuit in Figure~\ref{fig:w-prep}. We construct a size $2^k$ W-state codeword by first (a) constructing one of size 2 and then (b) recursively doubling its size. This circuit more closely resembles standard encoding procedures and uses a fixed number of non-Clifford gates; it also works to encode more than 1 qubit into the W-State code (c). It can also be used after Figure \ref{fig:elective-decoding} to reconstruct the W state code, akin to syndrome extraction.}
    \label{fig:ftencode}
\end{figure}

The circuits introduced in this section describe how to map a logical state $\ket{\psi}$ into the W-state code, following a structure inspired by the Schur transform~\cite{schur-transform}.
Due to their reliance on unitary projection operators, these circuits are susceptible to analog errors such as over- and under-rotations during implementation~\cite{analog}. 
While fault-tolerant encoding is, in principle, achievable to arbitrary precision through Solovay-Kitaev decomposition and magic state distillation for non-Clifford gates, an alternative approach based on SWAP and subspace controlled-NOT operations may be more practical to perform the mapping exactly using a fixed number of non-Clifford gates.
This design enables a clearer estimation of resource overheads and, as seen in Figure~\ref{fig:ftencode}, more closely resembles the structure of traditional encoders.

Fault-tolerant constructions for the qudit CSWAP and controlled-on-0 NOT gates have not yet been developed.
Both operations are non-Clifford qudit gates and require a thorough resource analysis of their own. 
Although fault-tolerant constructions for the Toffoli gate are well established~\cite{Toffoli}, these results do not generalize directly to the qudit setting.
The constructions presented here are therefore intended as a foundation for future analysis rather than a complete solution to the problem of qudit fault tolerance.
In Section~\ref{discussion}, we discuss how the use of more sophisticated codes, with established fault-tolerant implementations and superior performance characteristics relative to the W-state code, offers a promising path toward improved distributed error correction following the framework of this paper.

\section{Benefits of Approximate Codes for Distributed QEC} \label{adqec}

%In this section, we will explore how distributed quantum error correction with an approximate code (DAQEC) circumvents many of the difficulties faced by existing approaches to error correction for distributed quantum computing, as well as supplying additional benefits.
We propose the application of \emph{approximate} codes in distributed quantum error correction (DAQEC) as a novel approach to mitigate many of the substantial limitations of exact codes in the distributed setting, as discussed in Section~\ref{sec:dqec}.
We find that by circumventing the Eastin-Knill theorem, approximate codes may enable quantum computation with fewer processor-nonlocal gates than existing approaches without the need for magic states or code switching, in addition to exhibiting enhanced resilience under correlated noise sources.
We also find that approximate codes uniquely enable the composition of distinct QECCs across processors.
The benefits and drawbacks of DAQEC are summarized and contrasted against other approaches in Table~\ref{overview-table}.

\begin{table*}
\centering
\includegraphics[width=\textwidth]{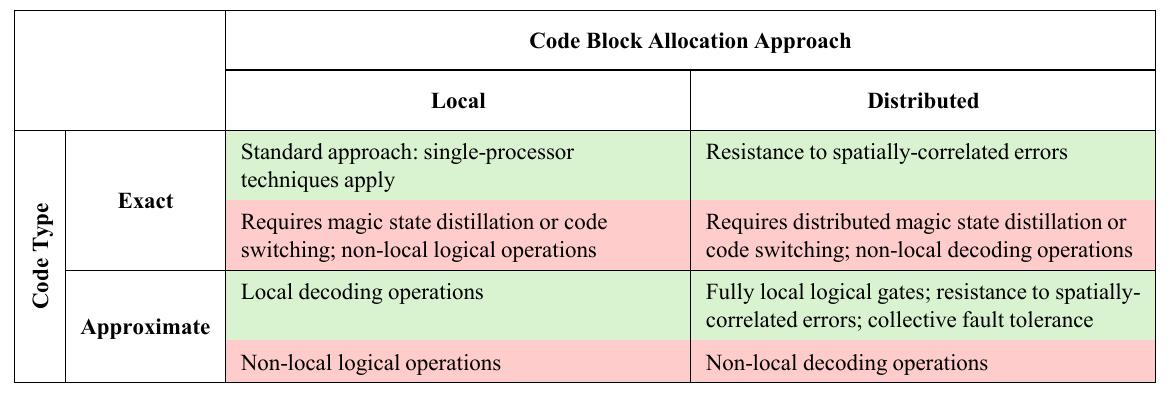}
\caption{Overview of the QEC approaches discussed in this paper in the distributed quantum computing context. Each approach carries advantages (shaded green) and disadvantages (shaded red). In particular, distributed approximate QEC provides properties not realized by other methods.}
\label{overview-table}
\end{table*}

\subsection{Bypassing the Eastin-Knill Theorem to Reduce Non-local Gates}

A fundamental theoretical advantage of approximate QECCs like the W-state code is their ability to evade certain no-go theorems that constrain exact codes. The Eastin-Knill theorem~\cite{eastin-knill} states that no quantum code can simultaneously (i) exactly correct all errors on each of its subsystems and (ii) realize a universal set of logical gates transversally. Therefore, any exact QEC code cannot admit a full set of transversal logical operations, requiring large-overhead schemes such as magic state distillation for universality in a distributed quantum computer~\cite{cost-of-universality}.

Approximate QECCs relax condition (i) by allowing some probability of logical error in the correction, and in return they may support a universal set of transversal gates. The W-state code is the simplest example of a code which meets these conditions. Recent works~\cite{rhea-alexander-approx-ekt, finite-approx-transversal, faist, kubica-approx-ek, universal-transversal-gates} have formalized approximate variants of the Eastin-Knill theorem, quantifying the trade-off between the worst-case error-correction infidelity $\varepsilon$ and the ability to admit transversal gates.

For distributed QEC, this is significant because it means that approximate codes can allow schemes that require only local operations on each module through transversal gates, without the need for intricate multi-QPU control during logical gate execution, at the cost of an (ideally small) logical error probability. In contrast, local and/or exact QEC codes require non-local control for logical operations when distributed across modules.

Due to this large overhead of exact schemes, here we consider the reduction of processor-non-local gates only for codes that admit universal sets of transversal gates.
Quantifying this reduction for a given QEC scheme compared to a baseline local scheme is a simple function of three variables: the number of processor-non-local 2-qubit physical gates required for encoding and decoding ($d_{\text{enc/dec}}$) and logical gate execution (for a single gate, this is equal to the code length $\ell_c$ in an $[[\ell_c, k=1, d\geq 3]]$ code multiplied by a ``nonlocality factor'' $\eta$), as well as the number of logical gates executed between each syndrome extraction or decoding stage ($d_\text{circuit}$).
For comparing a fully local scheme against a fully distributed scheme on $n_p$ processors, only the distributed $d_\text{enc/dec}$ and local $\eta=0$ are non-zero, so the distributed scheme sees an advantage when
\begin{equation}
d_\text{circuit}(1-\eta)\geq n_pd_\text{enc/dec}.
\label{eqn: basic}
\end{equation}
For more complex schemes --- e.g., when code blocks are neither fully distributed nor fully local --- both $d_\text{enc/dec}$ and $d_\text{circuit}$ may be non-zero. 
We now extend Equation~\ref{eqn: basic} to more general scenarios; given $n_L$ logical qubits, our analysis focuses on the case where $n_L=n_p$ using a code with length $\ell_c>n_p$.

%Consider a layout with $n_p$ processors, each able to store a logical qubit using some code with length $n>n_p$.
%The code block allocation which minimizes the number of non-local gates during logical gate execution is to allocate 
In the case that $n_p$ cleanly divides $\ell_c$, it is possible for all logical gates to be performed locally by allocating transversal partitions to the same processor. In cases where this division is not clean, this is still the optimal strategy (assuming all-to-all processor connectivity): allocate as many transversal partitions as possible to the same processors as evenly as possible ($q=\lfloor\frac{\ell_c}{n_p}\rfloor$) and partition the remainder ($s=\ell_c\text{ mod }n_p$) as symmetrically as possible.
Then repeat by partitioning the subsequent remainders in a similar fashion. The total number of required connections is given by $\frac{1}{2}\ell_cn_L(n_L-1)$, and the ratio of processor non-local gates after the partition is given by the nonlocality factor
\begin{equation}
\eta = \frac{s\left(n_p^2-ks^2-t^2\right)}{\ell_c n_L(n_L-1)}
\label{eqn: pnl}
\end{equation} 
where $k=\lfloor\frac{n_p}{s}\rfloor$ and $t=n_p\text{ mod }k$. 
Note that this expression assumes nothing about the first remainder $s$, but only holds for $s \text{ mod }t=0$. 
However, if $n_p<5$, $t$ cannot exceed $1$ for any odd $\ell_c$ (this would require $k=2$ and $n_p=4\implies s=2$, but $\ell_c\text{ mod }4$ cannot give an even result for odd $\ell_c$) and the expression above holds. 
It also holds for $\ell_c<11$ and for odd square $\ell_c\leq 625$ and $n_p<7$ (for odd valued square surface codes) by exhaustive checking. In general cases, $\eta$ can be simply bounded by
\begin{equation}
\eta \leq \frac{sn_p(n_p-1)}{\ell_cn_L(n_L-1)},
\end{equation} which in the cases where $n_L=n_p$ that we are considering is simply $s/\ell_c$ (this is the case where all remainder qubits are allocated as soon as they can be, which is optimal in some cases such as where $n_p=\ell_c-1$). Figure \ref{fig:pnl} shows an example with $n_p=3$ and $\ell_c=7$, which we can see also saturates this inequality. 
Then, assuming a uniform distribution of qubits involved in two-qubit gates across the encoder, decoder, and syndrome extraction, we can provide a bound as to where DAQEC provides advantage for $n_p<5$:

\begin{equation}
d_{\text{circuit}}(1-\eta)\geq d_{\text{enc/dec}}n_p\left(1-\frac{n_pq(q-1)}{\ell_c(\ell_c-1)}\right).
\end{equation} 
We can see that Equation~\ref{eqn: basic} follows as a special case.

\begin{figure}
\centering
\includegraphics[width=.89\columnwidth]{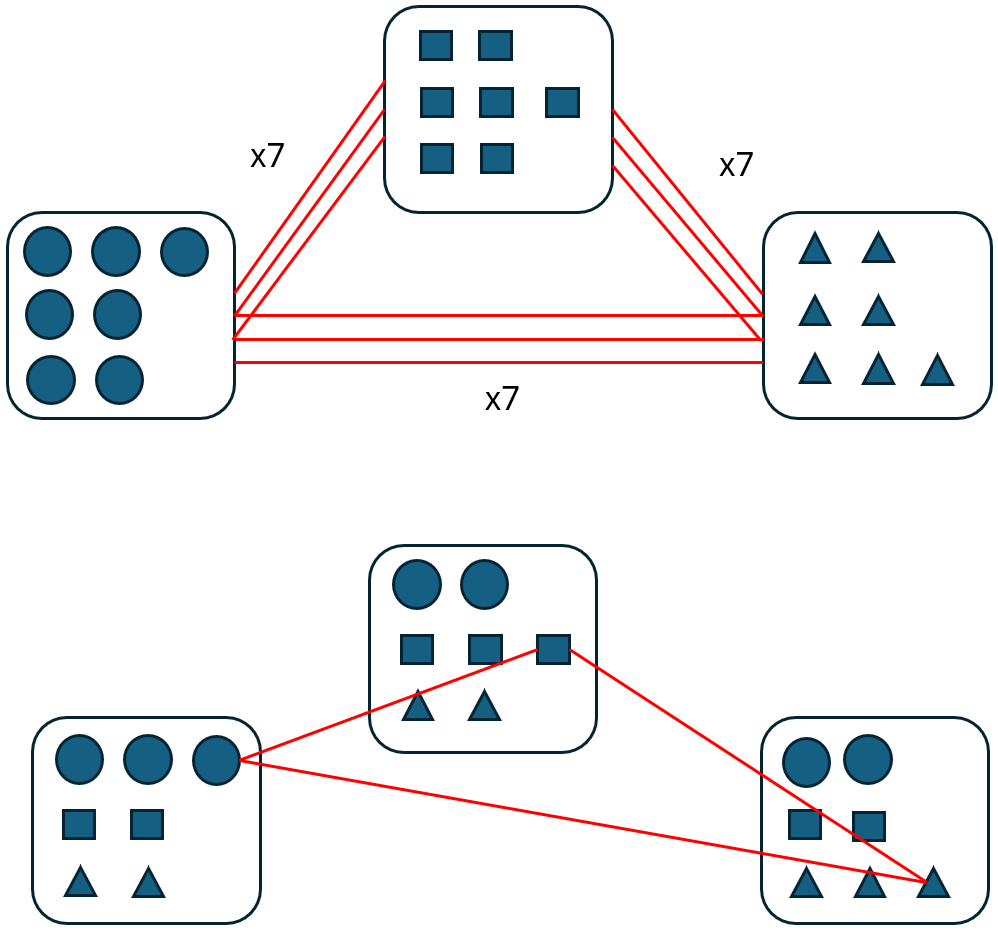}
\caption{Example non-uniform code block allocation. The different shapes denote three logical qubits, each encoded using a 7-qubit code. The red links highlight the processor-non-local gates required for transversal logical gate executions. Through DQEC allocation, the number of non-local gates is reduced from 21 to 3, in agreement with Equation \ref{eqn: pnl}.}
\label{fig:pnl}
\end{figure}

It is important to note that we can expect to allow many logical gates between each decoding step ($d_\text{circuit}>1$) for sufficiently low physical error rates when using DAQEC because errors do not propagate beyond one processor, i.e., in the worst case a single propagating error acts like a catastrophic event in the distributed case analyzed in~\cite{jiang-cosmic-ray-errors}.

Approximate QECCs like the W-state code which admit universally transversal gate sets therefore offer a path to distributed quantum computation without coupling the modules outside of QEC encoding and decoding, the costs of which may be amortized over many logical operations. 
The circumvention of the Eastin-Knill theorem therefore has the potential to translate to significant overhead and complexity reductions in a multi-QPU setting, without the need for magic state injection or code switching which are required by existing methods (see Section~\ref{existing-work}).

\subsection{Collective Fault Tolerance}

\begin{figure*}
    \centering
    \includegraphics[width=\textwidth]{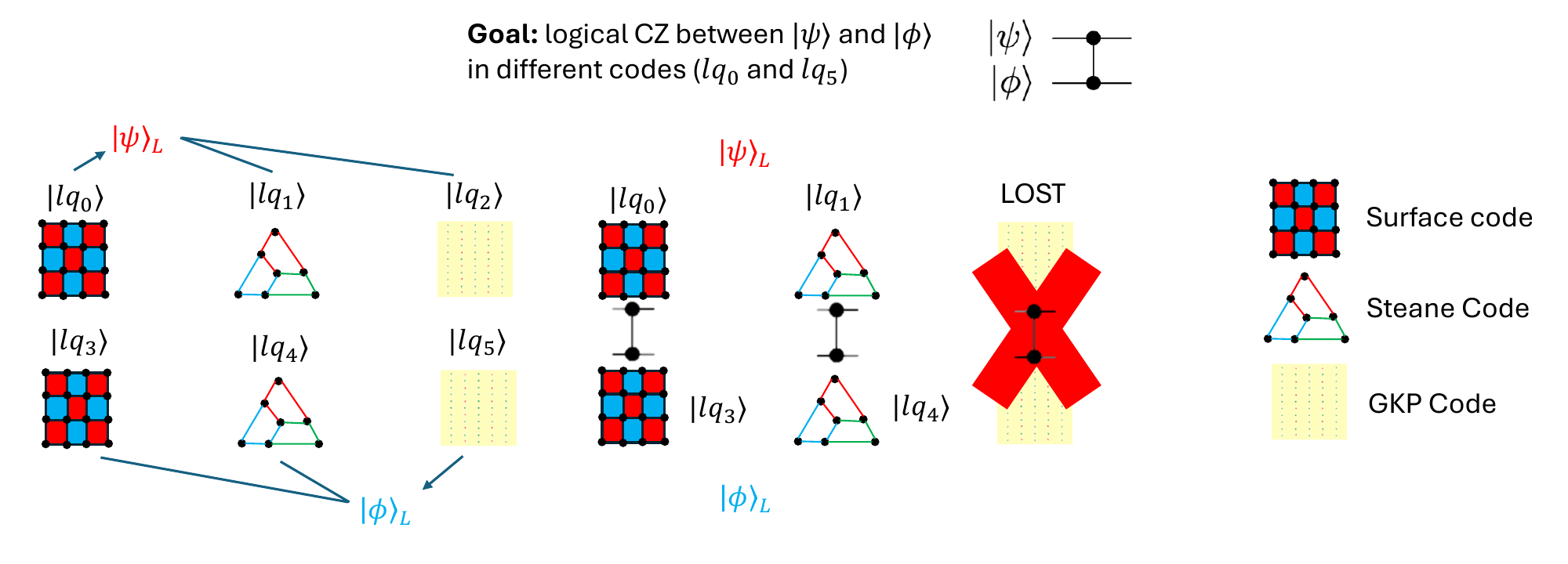}
    \caption{Example of collective fault tolerance in two logical qubits $\ket{\phi}_L$ and $\ket{\psi}_L$, each consisting of three heterogeneous codes concatenated with an approximate code. This concatenation scheme allows logical operations to be performed only amongst like codes, enabling distributed fault tolerance without high-overhead distributed protocols like magic state injection. If one code is erased, we can still recover with some probability using the approximate code.}
\label{fig:collective}
\end{figure*}

Codes that admit a universal transversal gate set enable a unique property in distributed quantum architectures that we call \emph{collective fault tolerance}.
This class includes approximate codes such as the W-state code as well as certain infinite-dimensional codes~\cite{GKP} and other constructions that evade the Eastin-Knill theorem~\cite{universal-transversal-gates}.
In such codes, all single- and two-qubit operations are transversal, meaning that any operation implemented fault-tolerantly within the constituent code blocks remains fault-tolerant when concatenated into an outer code.
This contrasts with conventional approaches to distributed fault tolerance, which require additional protocols such as distributed magic state injection.

This property is especially noteworthy when the outer code (approximate or otherwise) is concatenated with a set of heterogeneous inner codes -- for example, a surface code on one processor, a Steane code on another, and a Gottesman-Kitaev-Preskill (GKP) code on a third (see Figure~\ref{fig:collective}). 
In such a construction, logical operations on the distributed encoded qubits can be implemented fault-tolerantly by applying the appropriate local fault-tolerant operations, even if the underlying procedures differ.
For example, a non-Clifford gate might be realized via magic state distillation in the surface code~\cite{surface-msd}, code switching in the Steane code~\cite{steane-reedmuller-switching}, or with the standard GKP gadget in the GKP code~\cite{GKP}, and yet the concatenated qubit experiences the correct logical operation without any additional overhead.
This property also extends naturally to distributed settings where different processors, or distinct regions of a processor such as compute and memory zones~\cite{hetec}, employ different codes.
Because the collective logical operation is compiled from the constituent fault-tolerant operations, the overall system inherits fault tolerance from its parts.

A key observation is that this concatenated construction allows for the correction of certain errors that would be uncorrectable by either code alone, while allowing each processor-local code to operate independently.
To illustrate this concept, consider concatenating a stabilizer erasure-correcting code such as the $[[4,2,2]]$ code with a standard error correcting code such as the Steane code.
In this arrangement, the inner code can correct general errors, while the outer code can correct only erasures.
For a general $[[n,k,d]]$ code, the number of erasures $n_e$ and Pauli errors $n_p$ must satisfy $n_e+2n_p\leq d-1$ to be correctable.
A high-weight Pauli error on its own could cause a logical fault, but with physical error models such as those described in~\cite{quantum-circuit-inc}, the dominant logical error mode consists of a Pauli error coupled with an erasure that together exceed this threshold.
In this case, the presence of enough erasures signals that the error is unlikely to be corrected, and the Steane code block should be treated as erased by the outer code; this allows for correction of errors beyond the normal distance of the Steane code.
In regimes where Pauli errors dominate, one may instead employ a higher-distance code and adapt the decoding strategy so that errors are declared ``corrected'' only above some threshold of confidence.

The same principles apply when concatenating with a code admitting a universal transversal gate set: the distributed outer code can extend the effective error-correcting power of the system while maintaining locality of fault-tolerant gate implementations within each processor.
A full treatment of encoding and decoding strategies for this architecture is beyond the scope of this work; nonetheless, collective fault tolerance provides a mechanism for unifying heterogeneous codes into distributed logical qubits and enhancing the effective error resilience of the system as a whole.

\subsection{Advantage Under Noise Asymmetry}

Under approximate codes, distributed QEC maintains its advantage over local QEC in the face of spatially-correlated errors, in a similar manner to the case of exact codes discussed in Section~\ref{exact-distributed-spatially-correlated-errors}. We formally prove a lower bound on the magnitude of this advantage for the case of fully distributed codes against local codes in Theorem~\ref{thm:lower-bound} and validate this result numerically in Figure~\ref{fig:lower-bound}.

\begin{definition}[Fully Distributed QEC]
A QEC scheme in which each quantum processor contains no more than one subsystem assigned to any one code block.
\end{definition}

\begin{lemma}[Advantage of the Distributed W-state Code] \label{lem:ADQEC-adv}
Consider a modular platform consisting of $n$ processors each with $n$ qubits such that qubits on processor $p$ each admit an error with a probability $\varepsilon_p$. 
Suppose we assign $n$ W-state code blocks each of size $n$. Let $\varepsilon_{\text{dist}}$ be the probability of a logical fault across all blocks given a fully distributed assignment and let $\varepsilon_{\text{local}}$ be that under a local assignment.
Then $\varepsilon_{\text{dist}} \leq \varepsilon_{\text{local}}$.
\end{lemma}

\begin{proof}
Let the random variable $D_i$ be 1 if code block $i$ is decoded successfully and 0 otherwise. Similarly, let $D=\min(D_1,\dots,D_n)$ be 1 if all code blocks are decoded successfully and 0 otherwise. The probability of successfully decoding all blocks is then $\Pr[\text{decoding}] = \mathbb{E}[D] = \prod_{i=1}^n \mathbb{E}[D_i]$.

Let $X_i^p$ be 0 if qubit $i$ on processor $p$ experiences an error and 1 otherwise, and define $x_p:=1-\varepsilon_p$ to be the probability of no error for a qubit on processor $p$.  From Section~\ref{w-state-code} we deduce that the probability of successful decoding in the local and distributed cases are:
\begin{align}
\begin{split}
\mathbb{E}\left[D^\text{local}\right]
&= \prod_{p=1}^n \mathbb{E}\left[\frac1n \sum_{i=1}^n X_i^p\right] 
%= \prod_{p=1}^n \frac1n \sum_{i=1}^n \mathbb{E}\left[X_i^p\right] 
= \prod_{p=1}^n \frac1n \sum_{i=1}^n x_p \\
&= \prod_{p=1}^n x_p
\label{eqn:d-loc}
\end{split}
\\
\begin{split}
\mathbb{E}\left[D^\text{dist}\right]
&= \prod_{i=1}^n \mathbb{E}\left[\frac1n \sum_{p=1}^n X_i^p\right]
= \prod_{i=1}^n \frac1n \sum_{p=1}^n x_p \\
&= \left(\frac1n \sum_{p=1}^n x_p\right)^n
\label{eqn:d-dist}
\end{split}
\end{align}

By the AM–GM inequality, $\mathbb{E}\left[D^\text{dist}\right] \geq \mathbb{E}\left[D^\text{local}\right]$, hence $\varepsilon_{\text{dist}} \leq \varepsilon_{\text{local}}$.

\end{proof}

\begin{lemma} \label{lem:nth-root}
For $a\geq b>0$,
\begin{equation*}
a-b \geq n\sqrt[n]{b^{n-1}}\left(\sqrt[n]{a} - \sqrt[n]{b} \right)
\end{equation*}
\end{lemma}

\begin{proof}
\begin{align*}
a-b &= \left(\sqrt[n]{a} - \sqrt[n]{b} \right) \sum_{j=0}^{n-1} \sqrt[n]{a^{n-j-1}b^j} \\
&\geq \left(\sqrt[n]{a} - \sqrt[n]{b} \right) \cdot n\sqrt[n]{b^{n-1}}
\end{align*}
where the first equality follows by difference of powers.
\end{proof}

\begin{theorem}[Lower Bound on Distributed W-state Code Advantage] \label{thm:lower-bound}
Consider the $n$ processor setup of Lemma~\ref{lem:ADQEC-adv} and define $\varepsilon_{\text{dist}}, \varepsilon_{\text{local}}$ as above. Let $\sigma^2$ be the variance in the error probabilities of the processors. Then for low processor error rates, we can approximately lower bound the advantage of a fully distributed QEC scheme by 
\begin{equation}
\varepsilon_{\text{local}}-\varepsilon_{\text{dist}} \gtrapprox \frac{n\sigma^2}{2}
\end{equation}
\end{theorem}

\begin{proof}
By Lemma~\ref{lem:ADQEC-adv}, we know that $\varepsilon_{\text{dist}} \leq \varepsilon_{\text{local}}$. Then by Lemma~\ref{lem:nth-root}, we can write
\begin{align*}
    \varepsilon_{\text{local}}-\varepsilon_{\text{dist}}
    &= (1-\varepsilon_{\text{dist}}) - (1-\varepsilon_{\text{local}}) \\ 
    &\geq n(1-\varepsilon_{\text{local}})^{\frac{n-1}{n}} \left[(1-\varepsilon_{\text{dist}})^{\frac1n} - (1-\varepsilon_{\text{local}})^{\frac1n}\right]
\end{align*}

Observing that $1-\varepsilon_x = \mathbb{E}\left[D^x\right]$, we can use Equations~\ref{eqn:d-loc} and~\ref{eqn:d-dist} and a well-known approximation for the difference of arithmetic and geometric means ($\text{AM}-\text{GM} \approx\sigma^2/2$)~\cite{AM-GM-bound} to derive
\begin{align}
    \varepsilon_{\text{local}}-\varepsilon_{\text{dist}}
    &\geq n(1-\varepsilon_{\text{local}})^{\frac{n-1}{n}} \left[\mathbb{E}\left[D^\text{dist}\right]^{\frac1n} - \mathbb{E}\left[D^\text{local}\right]^{\frac1n}\right] \nonumber \\
    &\approx n(1-\varepsilon_{\text{local}})^{\frac{n-1}{n}} \frac{\sigma^2}{2} \nonumber \\
    &\geq n(1-\varepsilon_{\text{local}}) \frac{\sigma^2}{2} \label{eqn:bound-no-approx}
\end{align}

For low processor error rates, $1-\varepsilon_{\text{local}} \approx 1$ giving our stated bound.

\end{proof}

This inequality shows that the benefit of distributed error correction is proportional to the variance in processor error rates and the size of the system.
We numerically validate our lower bound in Figure~\ref{fig:lower-bound} for $n\in\{3,7,20\}$. Results show that the derived bound in Equation~\ref{eqn:bound-no-approx} is tight and the $1-\varepsilon_\text{local}\approx1$ approximation is valid for low physical error rates.

\begin{figure}
    \includegraphics[width=\columnwidth]{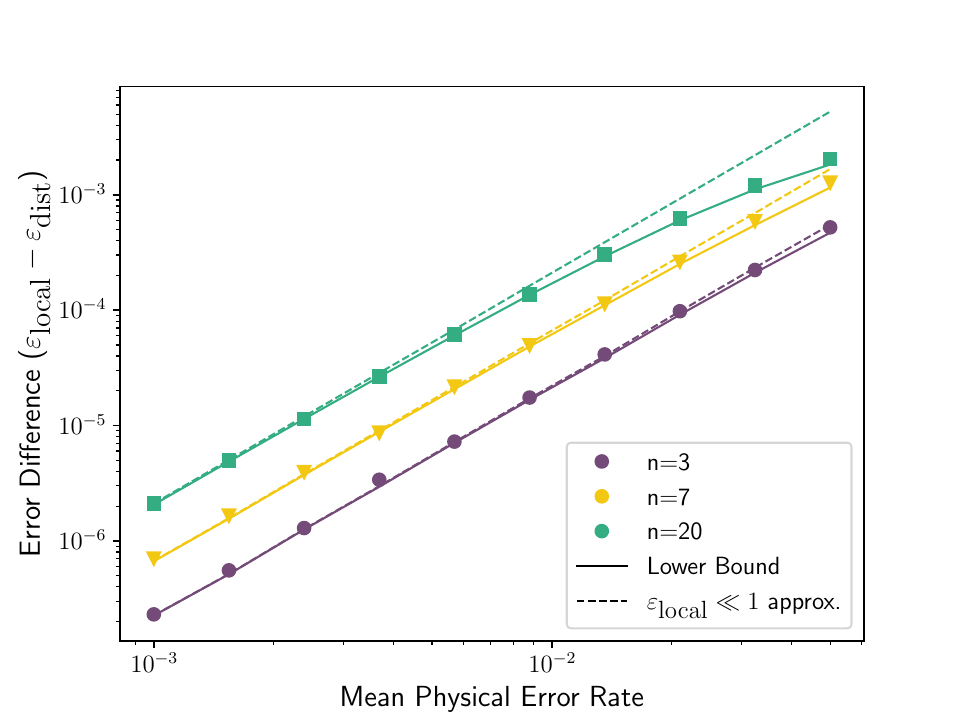}
    \caption{Numerical validation of Theorem~\ref{thm:lower-bound}. We simulate the relative performance of fully distributed and local QEC for $n$ processors of $n$ qubits for $n\in\{3,7,20\}$. The error difference between the approaches closely matches the bound in Equation~\ref{eqn:bound-no-approx} (solid lines), which is well approximated by $n\sigma^2/2$ (dashed lines) for low physical error rates.}
     \label{fig:lower-bound}
\end{figure}

Though this precise bound holds only for fully and uniformly distributed schemes with uniform error rates on each processor, the same principles extend to general spatially correlated errors, with advantage going to schemes which distribute errors across different code blocks. Sources of clustered errors include qubit crosstalk and cosmic ray events.

The advantage we derive here is not as large as that of the exact case (see Figure~\ref{dqec-vs-lqec}).
Intuitively, this is because error correction under exact codes is binary --- either there are sufficient physical errors to induce a logical fault, or there are not --- allowing for greater separation, whereas approximate codes have a probability of recovery that is linear in the number of errors, so the achieved separation is smaller.

%Accordingly, the loss threshold required to realize distributed advantage for physically relevant error rates is higher than in the exact case (\mytilde98\% vs \mytilde90\%) and there is not a significant advantage in logical error probability at those error rates; for more details, see Appendix \ref{sec:appendix}. It is worth noting that if realization efficiency is discounted, the distributed performance is still strictly better or equal to local performance at any error rate, and thus for circuits of sufficient depth a practical advantage is still realizable.

%\connor{It would be nice to generalize this result. Maybe to say that higher variance of error rates within code blocks (as opposed to between code blocks) is advantageous. And then argue that this generally results in a distributed code.}

\subsection{Realization}

To interconnect distributed QPUs in the manner described here, the first step is to prepare an approximate code state across those processors.
Techniques such as the W-state interconnects in~\cite{Lu} are suitable when qubit frequency mismatches between processors are small; for larger mismatches, frequency conversion may be required.
The QPUs can then embed the relevant qubits into a larger error correcting code of any type. 

One notable requirement of the W-state code in particular is that it is composed of qutrit subsystems.
However, this constraint is modest: in our construction only one physical subsystem from each local error correcting code must occupy a qutrit subspace, and all operations outside of W-state decoding may be performed in qubit subspaces.
Though not commonly used, qutrit degrees of freedom are available in superconducting circuits~\cite{scqt}, trapped ions~\cite{itqt}, and neutral atom processors~\cite{naqt}, and therefore present a feasible path to implementation.
This approach extends existing work~\cite{designation} supporting the use of dedicated communication qubits.

\section{Discussion and Future Work} \label{discussion}

In this work, we have explored the principles of DAQEC using the W-state code as a minimal yet illustrative example, providing a simple platform for demonstrating how universal transversality can more robustly tolerate noisy interconnects and enable collective fault tolerance in distributed architectures.
However, the W-state code represents only one instance within a broader class of codes that admit universal transversal gate sets.
These include infinite-dimensional codes~\cite{reference-frame-qec}, finite-dimensional approximate codes~\cite{finite-approx-transversal}, and multi-spin codes~\cite{universal-transversal-gates}. % Dicke state code??

Such codes offer desirable properties that surpass those of the W-state code, including higher encoding rates, improved scaling of error-correction infidelity, and known fault-tolerant constructions. 
Higher encoding rates allow these more advanced codes to use fewer physical qubits, while enhanced error resistance may enable them to tolerate greater interconnect losses.
Together, these features suggest that such codes could form the foundation of related schemes that are more compatible with near-term hardware.
Their greater complexity, however, places their analysis beyond the scope of this work.
%Finally, as mentioned in the main text, we leave an examination of fault-tolerant encoding and decoding (i.e. fault-tolerant constructions for controlled-on-0 not gates and qudit CSWAP gates) for future work.

The W-state code itself may also be generalized beyond the constructions considered in this work. While we focus on encoding qubit systems, the W-state code can in general encode qudits into qu$(d+1)$its, which may be desirable in some settings due to more efficient fault-tolerant non-Clifford gates \cite{campbell}, or better overall performance after concatenation. We leave an exploration of these ideas to future work.
%Finally, our construction treats methods of encoding qubits into this code. However, the W-state code can in general encode qudits into qu$(d+1)$its, which may have advantages in quantum computation such as more efficient fault-tolerant non-Clifford gates \cite{campbell}, providing better overall performance after concatenation. We leave an exploration of these ideas to a future work.

The results presented here establish the foundational concepts and potential advantages of DAQEC based on the transversality provided by approximate and other Eastin-Knill-evading codes.
Despite our focus on the simplest member of this class, the W-state code, which is suboptimal in both encoding rate and error-correcting capability and lacks a known fault-tolerant realization, we have shown that DAQEC can reduce non-local gate counts, mitigate spatially correlated errors such as crosstalk, and enable collective fault tolerance across distributed quantum processors. These represent substantial advances toward practical distributed quantum computation.
Future extensions of these ideas to more sophisticated codes promise to further enhance their utility and practicality.
\\ \:

\section{Acknowledgments}
The authors would like to thank Lane Gunderman, Zachary Guralnik, Rhea Alexander, Todd Brun, Yingkai Ouyang, Chris Gerhard, and Ching-Yi Lai for helpful discussions. We thank Leidos for funding this research through the Office of Technology. Approved for public
release 25-LEIDOS-0917-30170.
\bibliography{ref}

\appendix

\section{The ``Bad Apple'' Problem}
\label{sec:appendix}
\subsection{Proof that Optimal Packages are Fully Distributed}
Our model for distributed quantum error correction has a simple analogue in terms of apples. In the simplest case, suppose there are 3 bins each containing 3 apples. The first bin has apples which are 1 month old, the second has apples which are 1 week old, and the third has apples which are 1 day old. The probability that an apple taken from a given bin is expired is given by $p_m>p_w>p_d$ respectively. A crafty salesman wishes to divide these 9 apples into 3 barrels, each containing 3 apples, which he will then ship to customers. However, he wants to maximize the chance that no barrel has more than 1 expired apple --- in this case, TWO bad apples ruin the whole barrel (Figure ~\ref{apples}). Intuitively, he should pack the barrels with 1 apple from each bin; we can formalize this by considering the probability that the barrel is ruined:
\begin{equation}
    \begin{aligned}
        p_{\text{ruin}}&=p_1p_2(1-p_3)+p_1p_3(1-p_2)+p_2p_3(1-p_1)+p_1p_2p_3\\&=p_1p_2+p_1p_3+p_2p_3-2p_1p_2p_3\\&=\prod_{i=1}^3p_i\left(1+\sum_{j=1}^3\frac{1-p_{j}}{p_{j}}\right)
    \end{aligned}
\end{equation}
where $p_i$ is the probability that the $i$-th apple of the barrel is ruined. Similarly, the overall function we wish to maximize is $\prod_j(1-p_{j, \text{ruin}})$, i.e. the probability that none of the barrels are ruined, given the fixed set of $p$ values corresponding to our apples. More explicitly, we have
\begin{equation}
p_{\text{no bad barrels}}=\prod_k \left(\prod_i(1-p_{i,k})\left(1+\sum_j\frac{p_{j,k}}{1-p_{j,k}}\right)\right)
\end{equation}
\begin{equation}
=\prod_{i,k}(1-p_{i,k})\prod_{k}\left(1+\sum_j\frac{p_{j,k}}{1-p_{j,k}}\right)
\end{equation}
and the first product is fixed for any assignment of apples to barrels, since it is a product over all apples. So, we wish to maximize the second product only, subject to $\sum_{j,k} p_{j,k}=C$ or equivalently $\sum_{j,k}\frac{p_{j,k}}{1-p_{j,k}}=\sum_k F_k=C$. We use the method of Lagrange multipliers, optimizing over the logarithm:
\begin{equation}
\sum_k\log\left(1+\sum_j\frac{p_{j,k}}{1-p_{j,k}}\right)+\lambda\left(\sum_{j,k}\frac{p_{j,k}}{1-p_{j,k}}-C\right)=L
\end{equation}
\begin{equation}
\frac{\partial L}{\partial F_k}=\frac{1}{1+F_k}+\lambda=0\implies F_k=constant
\end{equation}
\begin{figure}
    \centering
    \includegraphics[width=\columnwidth]{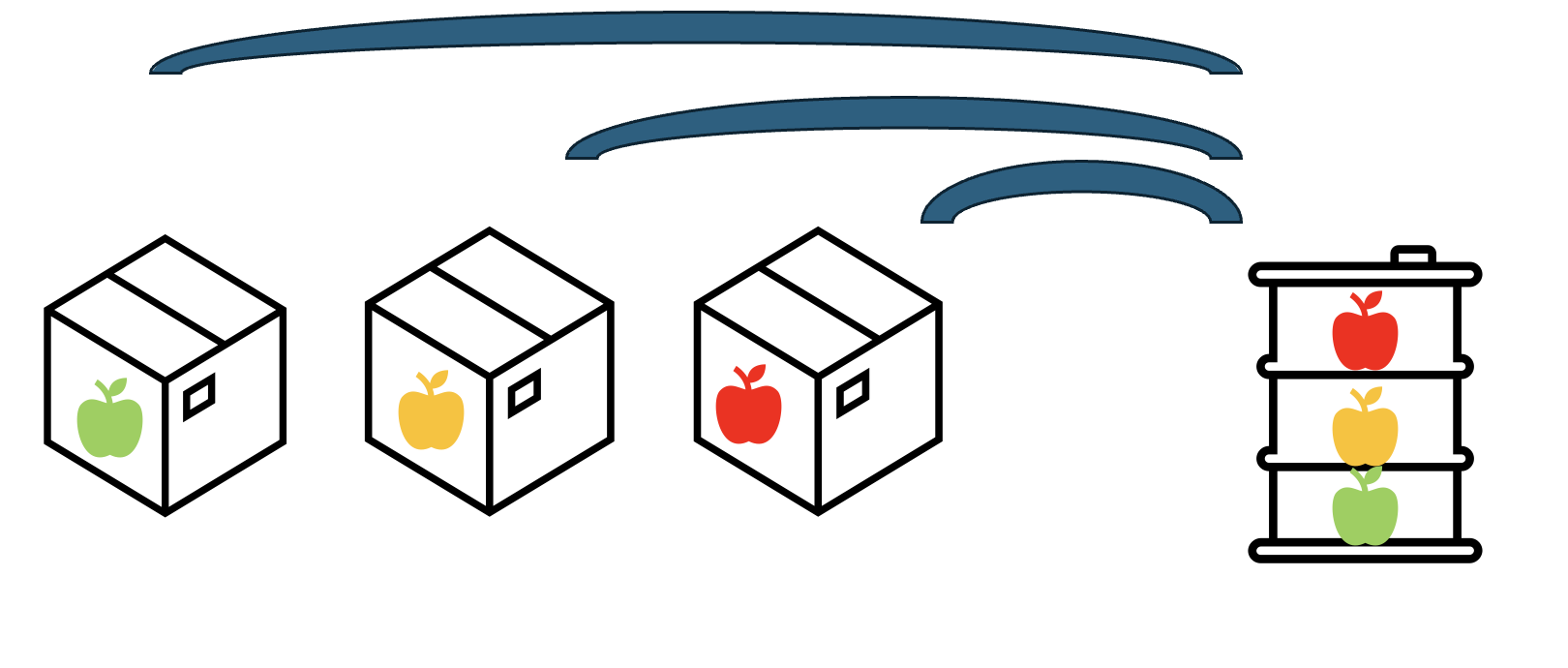}
    \caption{As the saying goes, one bad apple is...probably not enough to ruin the barrel if the other apples are really good.}
    \label{apples}
\end{figure}
and we see that $F_k$ has no $k$-dependence; therefore all the $F_k$ are equal and the barrel should be packed with one apple from each bin to minimize the probability of ruining any barrels. 

\subsection{Distribution Under Fixed Interconnect Loss Probability}

The above serves as an analogy for the simplest case, but suppose that when two apples from the same bin go into the same barrel, they can get contaminated in transit. This occurs with probability $p_c$, which is a constant (it does not vary based on which bins the apples came from). Then we must compare the probabilities of ruin between a heterogeneous barrel and a homogeneous barrel, each consisting of $n$ apples from $n$ bins. The latter is straightforward to calculate; the former is given by 
\begin{equation}
    p_{\text{ruin}}=1-(1-p_c)^{\binom{n}{2}}\left(\prod_{i=1}^n (1-p_i)\right)\left(1+\sum_{i=1}^n \frac{p_i}{1-p_i}\right)
\end{equation}
\begin{equation}
(1-p_\text{ruin})^n = \prod_i (1-(np_i(1-p_i)^{n-1}+(1-p_i)^{n}))
\end{equation}
\begin{equation}
    p_c^* \geq 1-\left(\frac{\left(\prod_i (1-(np_i(1-p_i)^{n-1}+(1-p_i)^{n}))\right)^{\frac{1}{n}}}{\left(\prod_i(1-p_i)\right)\left(1+\sum_i \frac{p_i}{1-p_i}\right)}\right)^{\frac{1}{\binom{n}{2}}}
\end{equation}
where we have considered all pairwise edges which may result in contamination, and this defines the cutoff threshold $p_c^*$ which can be evaluated numerically. For $p_m=.6$, $p_w=.2$, and $p_d=.05$, for example, the cutoff is $p_c^*\approx .07$.

Finally, we introduce the example corresponding to approximate QEC: suppose that rather than being binary, a barrel of size $n$ with $k$ expired apples is ruined with probability $k/n$ (which as an aside is a good metric for how bad a barrel of apples is, but defies intuition for error correcting code performance). If $p_i$ is the probability that apple $i$ is expired, then
\begin{equation}
p_{\text{ruin}}=1-\frac{(1-p_c)^{n-1}}{n}\sum_{i=1}^n (1-p_i)
\end{equation}  
for each barrel. We can then set the success probability in the fully contaminated case equal to that of the homogenous case:
\begin{equation}
(1-p_{\text{ruin}})^n=\prod_{i=1}^n (1-p_i)
\end{equation}
\begin{equation}
p_c^*\geq1-\left(\frac{n\left(\prod_i(1-p_i)\right)^{\frac{1}{n}}}{\sum_i(1-p_i)}\right)^{\frac{1}{n-1}}
\end{equation}
which gives the value of the cutoff probability assuming all links are utilized equally (maximally) in the creation of the code. More realistically, there will be uneven utilization, and so $p_\text{ruin}$ will be lower than anticipated. %\connor{Why does this not apply to the homogeneous case too? Because there is no link utilization in the homogeneous case}. 
This will make $p_c^*$ larger; therefore the above expression serves as a lower bound to the actual cutoff, and also serves as such in the exact case as well.

\end{document}